\documentclass[english]{IEEEtran}
\usepackage[T1]{fontenc}
\usepackage[latin9]{inputenc}
\usepackage{amsthm}
\usepackage{amsmath}
\usepackage{graphicx}
\usepackage{amssymb}
\usepackage{esint}
\usepackage{color}
\usepackage{flushend}
\usepackage[multiple]{footmisc}
\usepackage[draft,ulem=normalem]{changes}
\theoremstyle{plain}
\theoremstyle{plain}
\newtheorem{thm}{Theorem}
\newtheorem{lm}{Lemma}
\newcommand{\comment}[1]{}
\usepackage{babel}
\IEEEoverridecommandlockouts
\begin{document}

\title{{A Deterministic  Polynomial-Time  Protocol for Synchronizing from Deletions} \thanks{The paper is presented in part at the \textit{IEEE 7th International Symposium on Turbo Codes and Iterative Information Processing (ISTC)}, Aug. 2012.}}

\author{S. M. Sadegh Tabatabaei Yazdi and Lara Dolecek, \textit{IEEE Senior Member} \thanks{S. M. S. Tabatabaei Yazdi is with the Research and Development Center
at Qualcomm Inc., San Diego, CA, 92121. This work was done when he was a postdoc  
at the University of California, Los Angeles, CA, 90095. His email
address is {stabatab@qti.qualcomm.comm}. L. Dolecek is with the Department of Electrical Engineering at the University of California, Los Angeles, CA, 90095. Her email address is  dolecek@ee.ucla.edu.}
}
\maketitle
\begin{abstract}
In this paper, we consider a synchronization problem between nodes $A$ and $B$ that are connected through a two--way communication channel. {Node $A$} contains a binary file $X$ of length $n$ and {node $B$} contains a binary file $Y$ that is generated by randomly deleting bits from $X$, by a small deletion rate $\beta$. The location of deleted bits is not known to either node $A$ or node $B$. We offer a deterministic synchronization scheme  between nodes $A$ and $B$ that needs a total of  $O(n\beta\log \frac{1}{\beta})$  transmitted bits and reconstructs $X$ at node $B$ with probability of error that is exponentially low in the size of $X$. Orderwise, the rate of our scheme matches the optimal rate for this channel. 
\end{abstract}


\section{Introduction}
Consider two nodes $A$ and $B$ that respectively hold files $X$ and $Y$, where file $Y$ can be derived from file $X$ by some deletions.
For instance let 
\begin{align*}
X & =00\underset{D}{1}01\underset{D}{1}00\underset{D}{0}\underset{D}{1}0101\underset{D}{1}1, \mbox{and}\\
Y & =00010001011.\end{align*}
Here $Y$ is derived from $X$ by 5 deletions, 
where deleted bits are denoted by $D$.  We call $Y$ a deleted version of $X$. 

Suppose that the locations of deleted bits  are {\em unknown} to both nodes. In this paper we are interested in the following question:
\begin{itemize}
\item What is the optimal transmission protocol for {\em synchronizing} the content of node $B$ with the content of node $A$, i.e., how to reconstruct an estimate of file $X$ at node $B$? 
\end{itemize}
By way of optimality, we are mainly concerned with the number of transmitted bits between the two nodes and the complexity of implementing the protocol at nodes $A$ and $B$. Also, as usual, we desire the reconstructed estimate of $X$ at node $B$ to have bit error probability that is exponentially small in the size of $X$. 

Synchronization from deletions is a special case of a more general synchronization problem where file $Y$ can be derived from $X$ by a sequence of edits. An edit can refer to either deletion of a bit from  file $X$ or insertion of a new bit within $X$.    
File synchronization from random edits is the subject of many practical applications. Over the web, file updating is an application where a user or a server needs to synchronize its outdated version of a file with a newer version. The new updates of a file can usually be  modeled as random edits of its content. As another example, consider a search engine that constantly updates its database in order to reflect the latest changes to the content of websites. Here, as well, changes can be modeled by random edits to the content of websites. Another area of application is in distributed storage networks where several backup nodes store the same content and  need to be regularly synchronized together. Mis--synchronization in storage devices can be due to mis-synchronized clock speeds of read and write heads of hard drives or crashes in random parts of the hard drive. 

\subsection{Previous Work}
There has been a large body of research on {synchronization} from edits. In  \cite{Varshamov_Tenegolts},  Varshamov and Tenengolts  offered a coding scheme for recovery from one asymmetric error. Soon thereafter,  Levenshtein \cite{Levenshtein} showed that the scheme of Varshamov and Tenengolts can be used for synchronization from one deletion or one insertion. 
In \cite{Orlitsky}, Orlitsky proved several fundamental bounds on the minimum number of transmitted bits under a restricted number of communication rounds for a prescribed edit distance. 
While the results of \cite{Orlitsky} are nonconstructive, several researchers have provided explicit code constructions. Let $n$ denote the length of file $X$. For $\delta$ number of edits, 
Cormode {\em et al.} \cite{Cormode_Paterson_Sahinalp} offered an $\epsilon$-error protocol with $c(\epsilon) \delta \log^3 n$ total transmitted bits\footnote{All logarithms in this paper are in base $2$.},  where $c(\epsilon)$ is a constant that depends on the error $\epsilon$. For the same setting, Evfimievski \cite{Evfimievski} devised a protocol with the number of transmitted bits that is a polynomial in $\log n,\log\frac{1}{\epsilon},$ and $\delta$. For an unknown, fixed  number of edits  $\delta$, Orlitsky and Viswanathan \cite{Orlitsky_Viswanathan} showed that the $\epsilon$-error optimal protocol needs at most $\delta\log n+\log\frac{1}{\epsilon}$ transmitted bits. They also provided an explicit synchronization protocol that needs $2\delta \log n(\log n+\log\log n+\log\frac{1}{\epsilon}+\log \delta)$  {transmitted bits}. More recently, Venkataramanan {\em et al.} \cite{Venkataramanan_Zhang_Ramchandran} offered a synchronization scheme that can correct $\delta=o(\frac{n}{\log n})$ 
edits with $(4c+1){\delta\log n}$ transmitted bits from node $A$ to node $B$ and $10{(\delta-1)}$ transmitted bits from node $B$ to node $A$ for any positive integer $c$.  The error of  reconstruction is   at most  $\frac{d\log n}{n^c}$  where $d$ is the number of deleted bits in $X$, out of $\delta$ total edits.  

In practice, RSYNC \cite{Tridgell} is a popular UNIX application for synchronizing between edited files. The RSYNC method can be in general very inefficient and the number of transmitted bits can be exponentially larger than the optimal number. There have  been many improvements over the baseline approach. For example Suel {\em et al.} \cite{Suel_Noel_Trendafilov} proposed a protocol that in certain cases can save up to $50\%$ of bandwidth over RSYNC. There are also more specialized synchronization tools, such as VSYNC \cite{Zhang_Yeo_Ramachandran}, which synchronizes between video files.    

\subsection{Our Contribution}
While most of the previous work has concentrated on synchronizing from a fixed number of edits between two files $X$ and $Y$, in this paper we are interested in a more practical scenario, which is  synchronizing from a {\em fixed rate} of edits between two files. We only study synchronization from deletions, and will discuss possible extensions to the more general case of deletions and insertions at the end of the paper. More specifically, we consider synchronization between node  $A$ and  node $B$ where node $A$ has a binary string $X$ that is generated by an i.i.d. Bernoulli process of parameter $ \frac{1}{2}$. Node $B$ has a binary string $Y$ that is generated from $X$ by randomly and independently deleting bits of $X$ with probability $\beta$ that is very small. We are interested in an optimal transmission protocol for synchronizing between nodes $A$ and $B$ when $n$, the length of $X$, is large.

{We remark that,  throughout the paper,  by small $\beta$  we implicitly mean that there exists $\beta_0 > 0 $ such that our discussion is valid for all $\beta < \beta_0$. Furthermore, by large $n$ we implicitly mean that for every $\beta < \beta_0$ there exists a positive integer $n_{\beta}$ such that our discussion is valid for all $n > n_{\beta}$.}

In order to evaluate a lower bound on the optimal number of transmitted bits between nodes $A$ and $B$, suppose that node  $A$ has access to string $Y$. Then, the optimal number of transmitted bits to node $B$, needed for reconstructing $X$ is $H(X|Y)$, which is the conditional entropy of string $X$ given string $Y$. Ma {\em et al.} \cite{Ma_Ramachandran_Tse} considered a more general set-up where the deletion pattern follows a stationary Markov chain. By applying the result of  \cite{Ma_Ramachandran_Tse} to our model, for small values of $\beta$, the entropy $H(X|Y)$ can be estimated as follows 
\begin{equation}
H(X|Y) = n(\beta \log \frac{1}{\beta}+O(\beta)). 
\end{equation}
Therefore, any synchronization protocol needs  at least $n(\beta \log \frac{1}{\beta}+O(\beta))$ transmitted bits. Paper \cite{Ma_Ramachandran_Tse} further uses tools from the  well studied problem of  \emph {source coding with side information} \cite{Slepian_Wolf,Wyner_Ziv} to show that there exists a randomized synchronization protocol on a {\em one-way} channel that asymptotically needs $H(X|Y)$ transmitted bits. However, \cite{Ma_Ramachandran_Tse} does not offer any explicit, deterministic construction for the synchronization protocol. We remark that the most efficient previous constructions (e.g., \cite{Venkataramanan_Zhang_Ramchandran}) are for a {\em fixed}  number of edits $\delta$, and require $O(\delta \log n)$ transmitted bits between $A$ and $B$. A na\"{\i}ve application of such results to our setup would require  $O(n\beta \log n)$ transmitted bits between $A$ and $B$ for large $n$,  which is clearly far from being optimal.  

In this paper, we offer the first explicit and deterministic  construction of a protocol for synchronizing from a small rate of deletions on a two-way, error-free channel. The protocol is optimal {within} a constant multiplicative factor and needs $O(n\beta \log\frac{1}{\beta})$ transmitted bits. Furthermore, we demonstrate that the error probability of synchronization at node $B$ is exponentially small in $n$. Finally, we show that our scheme needs a running time that is at most $O(n^4\beta^6)$.   

The rest of the paper is organized as follows. In Section \ref{sec:Problem-Setting-and}, we present the  problem setting and the main result along with a sketch of our synchronization scheme. In Section \ref{sec:Coding-Scheme}, we present the mathematical details of our synchronization protocol and the proof of the main result in the paper. Section \ref{sec:Practical-Algorithms} discusses practical implications of our protocol for low--complexity synchronization algorithms, and Section \ref{conclusion} includes concluding remarks and directions for possible extensions. Preliminary results from this work were reported in~\cite{Sadegh_ISTC}.
\section{\label{sec:Problem-Setting-and}Problem Setting and the Main Result}
\subsection{Preliminaries}
We represent a binary string $Z$ of length $\ell$ by $Z=Z(1),Z(2),\cdots,Z(\ell)$.  For $1\leq i\leq j\leq \ell$, $Z(i,j)$ {denotes} the substring $Z(i),Z(i+1),\cdots,Z(j)$ of $Z.$ If $Z_1$ is a string of length $\ell_1$ and $Z_2$ is a string of length $\ell_2,$ we denote by $Z_1,Z_2$ the string of length $\ell_1+\ell_2$ obtained by concatenation of $Z_1$ and $Z_2.$ For a string $Z$, we let $|Z|$ denote the length of $Z$.

 \emph{Deletion channel} is a channel that may delete any subset of the bits of the input string. Let $X$ be the input to the deletion channel and $Y$ be the output of the channel. We  represent the set of deleted bits from $X$ by a binary  vector $D$ of length $|X|$ which is called the deletion pattern. If the deletion channel has deleted  bit $X(i)$ from $X$, then $D(i)=1$ and otherwise $D(i)=0.$ For example, the output of a deletion channel with input $X=101$ and deletion pattern $D=010$, is $Y=11.$

Corresponding to the deletion pattern $D,$ we define a function $f_{D}$ which maps the indices of bits in the input string, to their corresponding indices in the output string. If for index $i,$ $D(i)=0,$ then $f_{D}(i)=i-\sum_{j<i}D(i)$, and if $D(i)=1$, then $f_{D}(i)=f_{D}(i')$ where $i'$ is the largest index, smaller than $i$, for which $D(i')=0.$ In the example above $f_D(1) = 1, f_D(2)=1$, and $f_D(3) = 2.$
\subsection{The Main Result}
Suppose that {node $A$} contains a file that is represented by a binary string $X$ of length $n.$ Let {node} $B$ contain  file $Y$ of length $m$ that is the output of a deletion channel with input $X$ and deletion pattern $D$. We assume that the deletion pattern is unknown to nodes $A$ and $B$. Suppose that the source file $X$ is generated by an i.i.d. Bernoulli source of parameter $\frac{1}{2}$ and that the deletion channel  has deleted bits of $X$ independently and with probability $\beta\ll 1$. We are interested in a synchronization protocol on a  two-way, error-free channel between nodes $A$ and $B$ so that {node} $B$ can recover string $X$ from string $Y$ with a small probability of error at the end of the communication session. Our main contribution in this paper is proving the following theorem.
\begin{thm}
\label{thm:main-1} There exists a deterministic synchronization protocol between nodes $A$ and $B$ on a two-way, error-free channel, that on average transmits $O(n\beta\log\frac{1}{\beta})$ bits and generates an estimate $\hat{X}=\hat{X}(1),\cdots,\hat{X}(n)$ of $X$ at node $B,$ such that $\Pr\left\{ \hat{X}(i)\neq X(i)\right\} \leq2^{-\Omega(n)}$ for every $1\leq i\leq n$. 
\end{thm}
We prove the theorem by explicitly constructing a synchronization protocol. Next, we provide an overview of our synchronization protocol and prove its optimality. 
\subsection{Synchronization Protocol}
Recall that node $B$ has string $Y$ which is a deleted version of string $X$. We next explain a synchronization protocol that enables node $B$ to reconstruct an estimate of string $X$ with a small probability of error. 
The synchronization protocol has three main steps, as illustrated in Figure \ref{scheme}. Each step is performed by a module at node $B$ that has a two-way communication link to node $A$. The three modules work in series, such that the input to the first module is string $Y$ and the output of the last module is the estimate $\hat{X}$ of string $X$. 
\begin{figure*}
\centering
\includegraphics[scale=0.35]{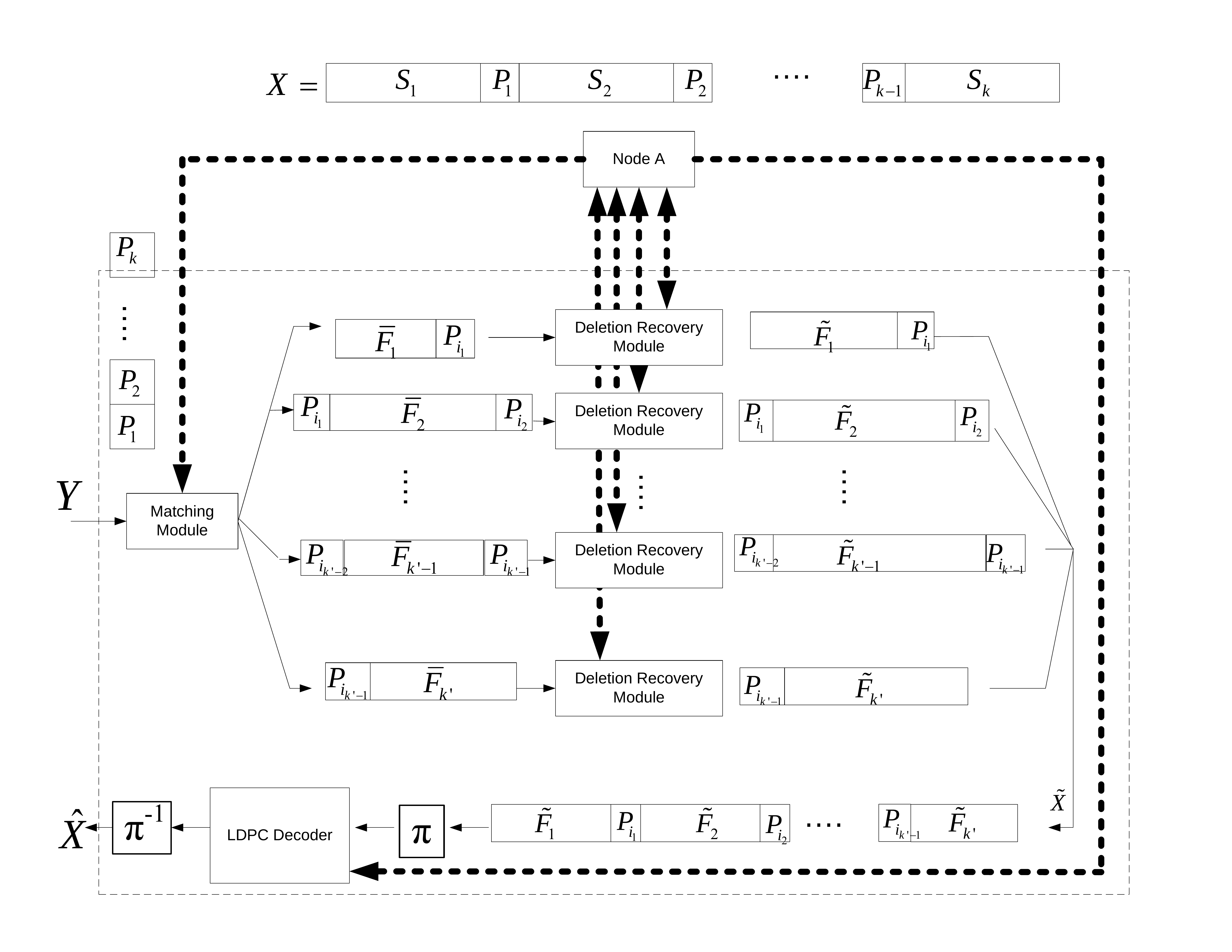}
\caption{\label{scheme}Illustration of the synchronization protocol.}
\end{figure*}
 Suppose that $X$ is partitioned into substrings as follows 
\[X=S_{1},P_{1},S_{2},P_{2},\cdots,S_{k-1},P_{k-1},S_{k},\]
 where $|P_{i}|=L_{P}$ and $|S_{i}|=L_{S}$. Substrings $P_{1},\cdots,P_{k-1}$ are called \emph{pivot strings} and substrings $S_{1},\cdots,S_{k}$ are called \emph{segment strings.} We choose $L_{S}=\frac{1}{\beta}$ and $L_{P}=O(\log\frac{1}{\beta})$ and both node $A$ and node $B$ know the exact values of $L_S$ and $L_P$. {Note that the length of a pivot string is much smaller than the length of a segment string.} We will determine the exact value of $L_P$ later during our analysis.
\begin{enumerate}
\item The first step of the synchronization protocol is performed by the \emph{matching module} at node $B$. In this step,  {node} $A$ sends pivot strings $P_i$, $1 \leq i \leq k-1$, in sequential order to {node} $B$. Upon receiving { all} the pivots, the matching module attempts to figure out the positions of pivots in $Y$ by finding the {\em exact copies} of $P_i$'s within $Y$. { The matching module is responsible for resolving ambiguities when there are multiple copies of a pivot in $Y$. The structure of the matching module and the graph-based algorithm for resolving the ambiguities are discussed in Section \ref{sec:Coding-Scheme}.} 
Due to {possible deletions} within $P_i$'s, the matching module is able to find the exact matches for only a subset of $P_i$'s. {We will explain later the other possible cases when there are multiple matches for a pivot but an error is made by detecting a match that is not due to the original pivot.}  Suppose that the matching module finds matches for $P_{i_1},\cdots,P_{i_{k'-1}}$ where $k'\leq k$.  Based on the position of matched $P_i$'s, the matching module partitions $Y$ into substrings as 
$$Y=\bar{F}_{1},P_{i_{1}},\bar{F}_{2},P_{i_{2}},\cdots,\bar{F}_{k'-1},P_{i_{k'-1}},\bar{F}_{k'},$$
 and sends this partitioned string to the next module,
where $\bar{F}_{j}$ denotes the substring between matched  pivots $P_{i_{j-1}}$ and $P_{i_{j}}$ in $Y$.
\item The next step is performed by the \emph{deletion recovery module} at node $B$. After receiving the partitioned $Y$ from the matching module, the deletion recovery module sends the indices $\{i_1,\cdots,i_{k'-1}\}$ of the matched pivots in $Y$ to node $A$. Upon receiving the indices, node $A$ partitions $X$ into substrings as follows:  
\begin{equation}
X=F_{1},P_{i_{1}},F_{2},P_{i_{2}},\cdots,F_{k'-1},P_{i_{k'-1}},F_{k'},
\end{equation}
where ${F}_{j}$ denotes the substring between pivots $P_{i_{j-1}}$ and $P_{i_{j}}$ in $X$. Substring $F_j$  can be written as follows:
\[{F}_{j}=S_{i_{j-1}+1},P_{i_{j-1}+1},\cdots,P_{i_{j}-1},S_{i_{j}}.\]
Notice that if $P_{i_{j-1}}$ and $P_{i_{j}}$ are matched correctly in $Y$, then $\bar{F}_j$ can be derived from $F_j$ by some {sequence} of deletions. 
In this step, nodes $A$ and $B$ use the synchronization protocol of Venkataramanan {\em et al.}, \cite{Venkataramanan_Zhang_Ramchandran} with parameter $c=3$  ($c$ is a parameter that defines the tradeoff between complexity of the protocol and the error in the output of the decoder) to recover from deleted bits of ${\bar F}_j$ and to form an estimate of $F_j$  for each $1\leq j\leq k$. Let us denote by $\tilde{F}_j$ the estimate of $F_j$ at the output of the deletion recovery module. Notice that $\tilde{F}_j$ has the same length as $F_j$. At the end of this step, the deletion recovery module forwards the string 
\begin{equation}
\tilde{X}=\tilde{F}_{1},P_{i_{1}},\tilde{F}_{2},P_{i_{2}},\cdots,\tilde{F}_{k'-1},P_{i_{k'-1}},\tilde{F}_{k'}
\end{equation}  
as an estimate of $X$ to the last module.
\item  At the last step, the  {\em LDPC decoder module} at node $B$, recovers from the errors made by the first two steps. Due to a potential existence of multiple copies of each $P_i$ within $Y$, the matching module (first step) may erroneously match $P_i$ at a wrong place. Suppose $P_{i_j}$ is a pivot that the matching module has {matched at a wrong place}. Then, $\bar{F}_j$ and $\bar{F}_{j+1}$ may not be realizable by deleting subsets of bits from $F_j$ and $F_{j+1}$ respectively. As a result, after the  deletion recovery module (second step), $\tilde{F}_j$ and $\tilde{F}_{j+1}$ may be different from $F_j$ and $F_{j+1}$,
respectively. Furthermore, even if the matching module has matched pivots $P_{i_{j-1}}$ and $P_{i_j}$ correctly in $Y$ and $\bar{F}_j$ is a deleted version of ${F}_{j}$, the protocol of Venkataramanan {\em et al.} \cite{Venkataramanan_Zhang_Ramchandran}, used in deletion recovery module, {could introduce additional errors.}

Suppose that the total error of the first two synchronization modules  is bounded by $\zeta$,  
\[\Pr\left\{\tilde{F}_j \neq F_j\right\} \leq \zeta.\]      
We notice that the output of the deletion recovery module, $\tilde{X}$, is in synchronization with $X$, in the sense that  $|\tilde{F}_j|= |F_j|$ for each $1\leq j\leq k'$ and hence ${\tilde X}(i)$ is the estimate of $X(i)$ for each index $1\leq i \leq n$. Since the  error rate over substrings $\tilde{F}_j$, $1\leq j\leq k'$, is an upper bound for the bit error rate over $X$, we find that 
\begin{equation}
\label{bsc_error}
\Pr\left\{\tilde{X}(i) \neq X(i)\right\}\leq \zeta.
\end{equation}
To recover from errors of $\tilde{X}$ we  use {a powerful additive-error correction code}. Our choice is an LDPC decoder which receives parity check bits of a systematic LDPC code \cite{Richardson_Urbanke_encoding}. { By applying a random permutation $\pi$ at the input of the LDPC decoder and its inverse permutation $\pi^{-1}$ at the output of the decoder, we can eliminate a potential non-uniformity of errors over different bits of $\tilde{X}$. 
Therefore, by using the error bound given in  \eqref{bsc_error},  the input sequence to the  LDPC decoder can be modeled as an output of a Binary Symmetric Channel (BSC) 
with a Bernoulli i.i.d. input sequence of parameter $\frac{1}{2}$ and
with the crossover probability of at most $\zeta$. We assume that node $A$ has access to  the permutation $\pi$.}

 If node $A$ sends a sufficient number of parity check bits to the LDPC decoder module, as shown in \cite{Richardson_Urbanke}, the output of the decoder will be a string $\hat {X}$ with
\[\Pr\left\{\hat{X}(i)\neq X(i)\right\} \leq 2^{-\Omega(n)},\]
{as previously stated in Theorem \ref{thm:main-1}.}
\end{enumerate}
Next, we wish to estimate the total number of transmitted bits {used} by our synchronization protocol. {We first establish} a measure of the performance of the matching module of the decoder.  
\begin{thm}
\label{matching_theorem}
Let $k'=(1-L_P\beta +2\beta + o(\beta))k$. For $L_P\geq 11+2\log\frac{1}{\beta}$, there exists a matching module  that with probability $1-2^{-\Omega(n)}$, matches a subset $\{P_{i_1},\cdots,P_{i_{k'-1}}\}$ of pivots $\{P_1,\cdots,P_{k-1}\}$ such that the {probability of} error  in matching  $P_{i_j}$ is  at most  $\beta + o(\beta)$.
\end{thm}
We devote Section \ref{sec:Coding-Scheme} to proving this theorem. For the rest of our argument we set $L_P = 11+2\log \frac{1}{\beta}$, which is  the minimum value of $L_P$ required by Theorem \ref{matching_theorem}. 

Next, we use Theorem \ref{matching_theorem} to estimate the total number of transmitted bits needed by the synchronization protocol.
\begin{lm}
On average, the total number of transmitted bits of the {synchronization} protocol is no more than $109n\beta \log \frac{1}{\beta}$.
\end{lm}    
\begin{proof}
First notice that $k=\frac{n+L_P}{L_S+L_P} =  n\beta + 11 \beta + 2\beta \log \frac{1}{\beta} = n\beta + o(1).$ 
The number of transmitted bits in the first step is \[(k-1)L_P =2n\beta \log\frac{1}{\beta}+o(n\beta \log\frac{1}{\beta}).\]
At the second step, node $B$ needs no more than $k = n\beta$ bits to transmit the indices $\{i_1,\cdots, i_{k'-1}\}$ to node $A$. 
Furthermore, the protocol of Venkataramanan  {\em et al.} \cite{Venkataramanan_Zhang_Ramchandran}
for the recovery from deletions within each $F_j,1 \leq j \leq k'$, with parameter $c=3$  needs $13\delta_j \log |F_j|+10(\delta_j -1)$ transmitted bits, where $\delta_j := |F_j|-|\tilde{F}_j|$ is 
the number of deleted bits in $F_j$.  Therefore, the average number of transmitted bits in the second step is no more than
 \[n\beta+
\mathbb{E}\left[\sum_{j=1}^{k'}\left(13\delta_{j}\log|F_{j}|+10\delta_{j}\right)\right].\]
 Notice that $\sum_{j=1}^{k'}\delta_{j}$ is the total number of deleted
bits from $X$ and is on average $n\beta$ {(recall that we assumed that no deletions occurred in the matched pivots).}

In 
Appendix I we show that $\mathbb{E}\left[\delta_{j}\log|F_{j}|\right]\leq16+8\log\frac{1}{\beta}.$
Therefore, the average number of transmitted bits in the deletion recovery module is upper bounded
by \begin{align*}
& n\beta+k' \cdot 13(16+8\log\frac{1}{\beta})+10n\beta \\ \leq & n\beta \cdot 13(16+8\log\frac{1}{\beta})+11n\beta
 \\ = & 104 n\beta\log\frac{1}{\beta} + o(n\beta\log\frac{1}{\beta}),\end{align*}
where we used the inequality $k'\leq k = n\beta + o(1).$ 

For the last step, we would like to estimate the error $\zeta$ in $\tilde{F}_j$. By Theorem \ref{matching_theorem}, the error {probability} in matching $P_{i_{j-1}}$ and $P_{i_{j}}$  is 
at most $\beta + o(\beta)$ {each}. Since $\bar{F}_j$ is the common neighbor of $P_{i_{j-1}}$ and $P_{i_{j}}$, with probability at most $2\beta + o(\beta)$, the string $\bar{F}_j$ is not a deleted version of $F_j$. Furthermore, the 
error in the protocol of Venkataramanan {\em et al.} \cite{Venkataramanan_Zhang_Ramchandran} for $c=3$, is upper bounded by $\frac{\delta_j \log|F_j|}{|F_j|^3}$. Since ${\mathbb E} \left[\delta_j\right]=\beta L_S=1$ and also $|F_j|=L_S=\frac{1}{\beta}$,
the average probability of error by the protocol of Venkataramanan {\em et al.},  is {upper}bounded by $\beta^3 \log\frac{1}{\beta}=o(\beta)$. 
Counting the error from the matching module, we have $\Pr\left\{\tilde{F}_j \neq F_j\right\}\leq 2\beta+o(\beta)$, and therefore $\Pr\left\{\tilde{X}(i) \neq X(i)\right\}\leq 2\beta+o(\beta)$. 

In order to recover from errors induced by a BSC with crossover probability {of at most}  $2\beta+o(\beta)$, node $A$ needs to send \[nH(2\beta+o(\beta))= 2n\beta \log \frac{1}{\beta} + o(n\beta \log\frac{1}{\beta}),
 \] parity check bits to node $B$,  where we use $H(\cdot)$ to refer to the binary entropy function defined as $H(t)=t\log \frac{1}{t} + (1-t) \log \frac{1}{1-t}$ for $0<t<1$. 

The average number of transmitted bits in all three steps of the protocol is { upper bounded by $108 n\beta\log\frac{1}{\beta} + o(n\beta\log\frac{1}{\beta}) <109 n\beta\log\frac{1}{\beta}$.}  Therefore, the average number of transmitted bits by the algorithm is { no more than $109 n\beta\log\frac{1}{\beta}.$}  
\end{proof}    
In the next section we prove Theorem \ref{matching_theorem}.
\section{\label{sec:Coding-Scheme}Proof of Theorem \ref{matching_theorem}}
In this section,  we  propose a construction of a matching module such that for $L_P \geq 11+2\log\frac{1}{\beta}$, with probability $1-2^{-\Omega(n)}$,  the module matches $k'$ pivots,   out of  which at most $\beta k$ pivots are matched erroneously. Since ${\beta k}=(\beta+o(\beta))k',$ our construction implies an error of at most $\beta+o(\beta)$  in matching the pivots. This claim  is equivalent to the statement of  Theorem \ref{matching_theorem}.
   
We will frequently use the following concentration theorem  in our argument:
\begin{thm}[Hoeffding \cite{Hoeffding}]
\label{Hoeffding}
Let $p_{0}$ be the probability that a biased coin shows heads. Then for every $\varepsilon >0,$
the probability that $N$ tosses of the coin yield a number of heads between $(p_{0}-\varepsilon)N$
and $(p_{0}+\varepsilon)N$ is at least $1-2e^{-2\varepsilon^{2}N}.$ 
\end{thm}
We will occasionally need a stronger version of the previous theorem:
\begin{thm}[Hoeffding \cite{Hoeffding}]
\label{Hoeffding_2}
Let $z_{1},\cdots,z_{N}$ be i.i.d. random variables with expected
value $M$ that take values in an interval of length $I.$ Then, for every $\varepsilon > 0$, the
following holds\[
\Pr\left\{\left|\sum_{i=1}^{N}z_{i}-NM\right|\geq\varepsilon N\right\}\leq2\exp \left(-\frac{2\varepsilon^{2}N}{I^{2}}\right).\]
\end{thm}
 Recall that  string $X$ is partitioned into substrings as $X=S_{1},P_{1},\cdots,S_{k-1},P_{k-1},S_{k},$
where $|S_{i}|=L_{S}$ and $|P_{i}|=L_{P}.$ In our set-up, $L_{S}=\frac{1}{\beta},L_{P}=O(\log \frac{1}{\beta}),$
and $k = n\beta + o(1).$  Let us denote the index of the first bit of $P_{i}$ in $X$
by $\check{p}_{i}$ and the index of the last bit of $P_{i}$ in $X$ by $\hat{p}_{i}$.
Similarly, the first and last indices of $S_{i}$ are denoted by $\check{s}_{i}$
and $\hat{s}_{i}$. Therefore, $X(\check{p}_{i},\hat{p}_{i})=P_{i}$
and $X(\check{s}_{i},\hat{s}_{i})=S_{i}.$

The task of the matching module is to find  ``correct matches'' of $P_i$'s within string $Y$. Next, we formalize the notion of 
correct and incorrect matches for a pivot $P_i$. 
\subsection{Correct and Incorrect Matches}
Consider the substring $D(\check{p}_{i},\hat{p}_{i})$ which is the  part of the deletion pattern $D$ that acts on the pivot $P_i$. We consider the following cases:
\begin{itemize}
\item {\em $D(\check{p}_{i},\hat{p}_{i})$ is the all zeros vector:} There is no deletion within $P_i$. In this case we call
the copy of $P_i$ between indices $f_D(\check{p}_{i})$ and $f_D(\check{p}_{i})$ of $Y$ the {\em correct match} of $P_i$. All other copies
of $P_i$ in $Y$ are considered {\em incorrect matches} of $P_i$.
\item  {\em $D(\check{p}_{i},\hat{p}_{i})$ has one nonzero element:} There is one deletion within $P_i$. In this case, if there is a 
copy of $P_i$ in $Y$ that begins at $f_D(\check{p}_{i})$ or ends at $f_D(\hat{p}_{i})$ then we call it a correct match of $P_i$ and all other copies of 
$P_i$ are called incorrect matches of $P_i$. If there is no such copy of $P_i$ within $Y$, then all copies of $P_i$ within $Y$ are called incorrect matches. 
Notice that in this case there are possibly two correct matches for $P_i$. For instance, let $P_i=000$ and let the immediate undeleted bits before and after 
$P_i$ be zero. Then it is easy to verify that after one deletion within $P_i$, there is a copy of $P_i$ starting at $f_D(\check{p}_{i})$ in $Y$ and 
there is another copy of $P_i$ ending at $f_D(\hat{p}_{i})$ in $Y$.
\item {\em $D(\check{p}_{i},\hat{p}_{i})$ has more than one nonzero element:} There is more than one deletion within $P_i$. In this case all
copies of $P_i$ within $Y$ are considered incorrect matches.
\end{itemize}
\begin{figure}
\centering
\includegraphics[scale=0.08]{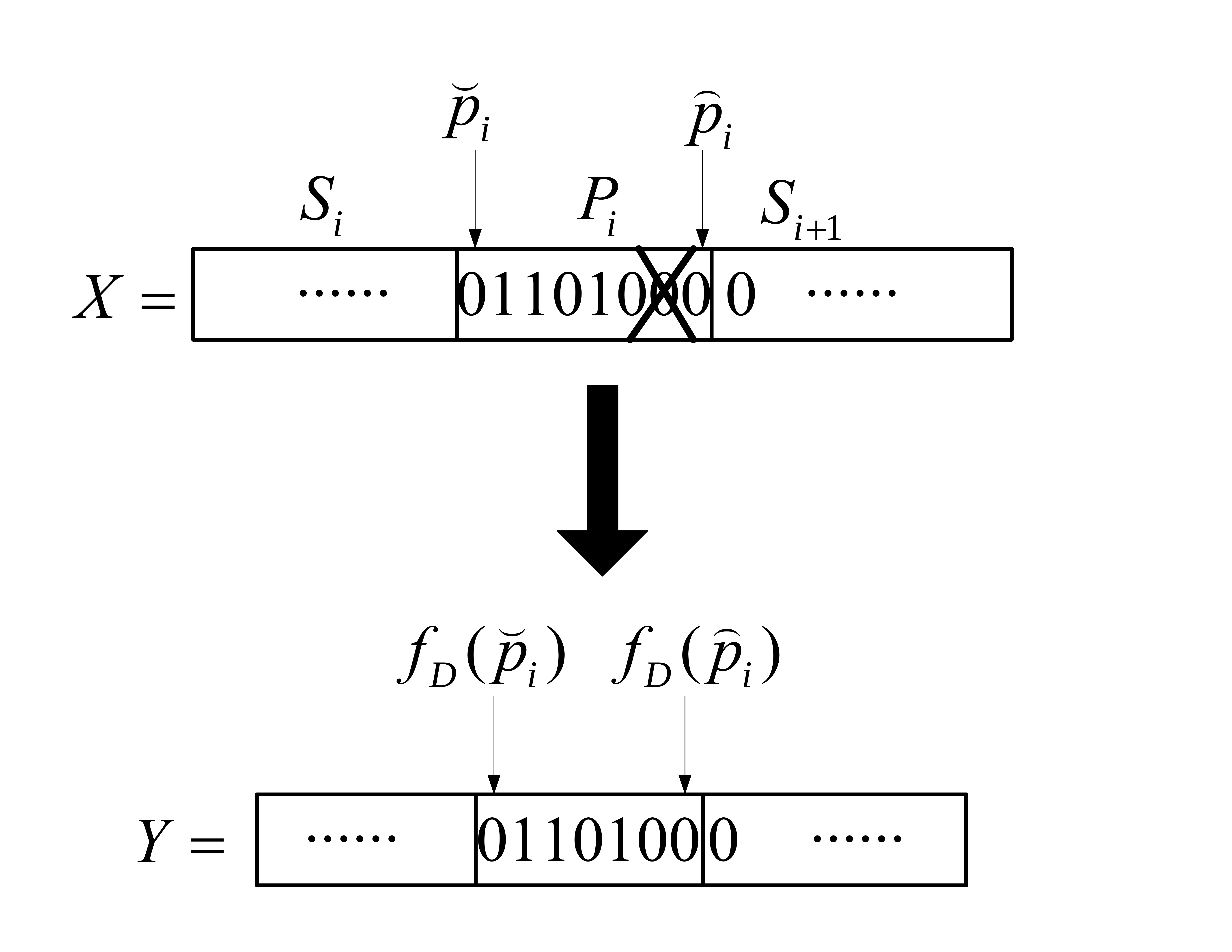}
\caption{\label{one_deletion}Illustration of a correct match of $ P_i$ with one deletion.}
\end{figure}
While the definition of correct and incorrect matches is natural for the case of no deletion within $P_i$, we next explain the reason behind
the definition for the case with deletions within $P_i$. Consider the illustration in Figure \ref{one_deletion} where $P_i=01101000$. Assume the penultimate bit
is deleted from $P_i$. Suppose that the bit right after $P_i$ is $0$. Notice that even with the deleted bit, a copy of $P_i$ appears in $Y$, starting 
at  $f_D(\check{p}_{i})$. This copy of $P_i$ is called a correct match. The reason is that the resulting string $Y$ is the same as in the case where
there is no deletion within $P_i$ and instead the $0$ after $P_i$ is deleted in $X$. In other words, here we can ``move'' the deletion from
$P_i$ to the substring $S_{i+1}$ without changing $Y$. 
 
Although a similar scenario
may happen when there are more than one deletions within $P_i$, i.e., we might be able to move the deleted bits from $P_i$ to the neighboring
segment strings without changing the resulting $Y$, since the probability of these cases is very small (the exact statement will follow), our analysis conservatively counts those matches as incorrect matches.       

Next, we analyze the probability of occurrence of correct matches for $P_i$:
\begin{lm}
\label{no_deletion_lemma}
With probability $1-\beta L_P + o(\beta)$, $P_i$ has no deletions and there is at least one correct match for $P_i$ within $Y$. 
\end{lm}
\begin{proof}
With probability $(1-\beta)^{L_P}$ no bit is deleted from $P_i$. For $L_P=O(\log\frac{1}{\beta})$ we have 
\[(1-\beta)^{L_P} = 1-L_P\beta+o(\beta).\]\end{proof}
\begin{lm}
\label{one_deletion_lemma}
With probability $2\beta+o(\beta)$ there is one deletion within $P_i$ and there is a correct match for  $P_i$ within $Y$. 
\end{lm}
\begin{proof}
Fix $h$ as the place of the deleted bit out of $L_P$ bits of $P_i$. Suppose $P_i(h)=b \in \{0,1\}$. It is simple to observe that there is a copy of $P_i$ starting at $f_D(\check{p}_{i})$ in $Y$
if and only if $P_i(h,L_P)=b,b,\cdots,b$  and furthermore, the first undeleted bit 
after $P_i$ in $X$ is also $b$. In other words, the $h$th bit of $P_i$ should belong to the final ``run'' of zeros or ones of $P_i$ and the first undeleted bit after $P_i$ should also be of the same value. 
With probability $\beta(1-\beta)^{L_P-1}$, exactly the $h$th bit of $P_i$ is deleted and with probability $2^{-(L_P-h+1)}$ the bits after $h$th bit in $P_i$ and the first undeleted bit after $P_i$ have the same value     
as the $h$th bit of $P_i$. The overall probability of this case is $\beta(1-\beta)^{L_P-1}2^{-(L_P-h+1)}$. Similarly, there is a copy of $P_i$ finishing at $f_D(\hat{p}_{i})$ in $Y$ if and only if all bits before the $h$th bit in 
$P_i$ and the first undeleted bit before $P_i$ are equal to the $h$th bit of $P_i$. This case happens with probability $\beta(1-\beta)^{L_P-1}2^{-h}$. The intersection of the two events happens when $P_i$ is all-zeros or all-ones string and the immediate undeleted bits before and after $P_i$ have the same value as the bits in $P_i$. This case happens with probability $\beta(1-\beta)^{L_P-1}2^{-(L_P+1)}$. By using the inclusion-exclusion principle and by varying $h$ from $1$ to $L_P$, we find the total probability of having one deletion within $P_i$, and as a result at least one correct match for $P_i$ to be:
\begin{align*}
\beta(1-\beta)^{L_{P}-1}\sum_{h=1}^{L_{P}}\left(2^{-(L_{P}-h+1)}+2^{-h}-2^{-(L_{P}+1)}\right) & =\\
\beta(1-\beta)^{L_{P}-1}(2-2^{1-L_{P}}-L_{P}2^{-(L_{P}+1)}) & =\\
2\beta+o(\beta),\end{align*} 
where in the last step we assumed $L_P=O(\log\frac{1}{\beta})$.
\end{proof}
\begin{lm}
\label{multiple_deletions_lemma}
With probability $o(\beta)$, $P_i$ has more than one deletion. 
\end{lm}
\begin{proof}
Since the probability of no deletion within $P_i$ is $(1-\beta)^{L_P}$ and the probability of one deletion within $P_i$ is $L_P\beta(1-\beta)^{L_P-1}$,
then the probability of more than one deletion within $P_i$ is 
\[1-(1-\beta)^{L_P}-L_P\beta(1-\beta)^{L_P-1}=o(\beta),\]
where we assumed $L_P=O(\log\frac{1}{\beta})$ in the final estimate.
\end{proof}
Let us define 
\begin{equation}
\label{R}
R:=1-L_P\beta+2\beta.
\end{equation}
 From the preceding lemmas we conclude that: 
\begin{lm}
For a random string $X$ and a random  deletion pattern $D$, on average, the number of pivots with at least one correct match in $Y$ is $(R+o(\beta))k$.  
\end{lm}
By applying Theorem \ref{Hoeffding} we conclude that:
\begin{lm}
\label{good_vertices}
For a random string $X$ and a random deletion pattern $D$, with probability $1-2^{-\Omega(n)}$, there are $(R+o(\beta))k$ pivots with at least one correct match in $Y$.  
\end{lm}
\begin{proof}
The probability that a pivot has a correct match in $Y$ is $R+o(\beta)$ and it is independent of other pivots. Therefore, if in  Theorem \ref{Hoeffding} we set $p_0$ to $R+o(\beta)$,  $N$ to $k$, and $\varepsilon$ to $o(\beta)$, we conclude that the probability that for  a random string $X$ and a random string $D$ there are between $(p_0 - o(\beta))k$ and $(p+ o(\beta))k$   pivots with correct matches in $Y$, is at least $1-2e^{-2o(\beta)^2k} = 1-2^{-o(\beta) n} = 1-2^{-\Omega(n)}$. The fact that the set of  integers between $(p_0 - o(\beta))k$ and $(p+ o(\beta))k$ can be represented by the set of integers of the form $(R+o(\beta))k$, yields the result.  
\end{proof}
\begin{lm}
\label{two_matches}
For a random string $X$ and a random deletion pattern $D$, with probability $1-2^{-\Omega(n)}$, there are $o(\beta)k$ pivots with two correct matches in $Y$.  
\end{lm}
\begin{proof}
As we showed in the proof of Lemma \ref{one_deletion_lemma}, the probability that a pivot has a deletion and two correct matches in $Y$ is given by the following expression
\[L_P\beta (1-\beta)^{L_P-1}2^{-(L_P+1)} = o(\beta), \]
where we assumed $L_P = O(\log\frac{1}{\beta}).$  Therefore, the average number of pivots with two correct matches in $Y$ is $o(\beta)k.$ Now, if in  Theorem \ref{Hoeffding} we set $p_0$ to $o(\beta)$,  $N$ to $k$, and $\varepsilon$ to $o(\beta)$, we conclude that the probability that for  a random string $X$ and a random string $D$ there are between $(p_0 - o(\beta))k$ and $(p+ o(\beta))k$   pivots with two correct matches in $Y$, is at least $1-2e^{-2o(\beta)^2k} = 1-2^{-o(\beta) n} = 1-2^{-\Omega(n)}$. Since the set of  integers between $(p_0 - o(\beta))k$ and $(p+ o(\beta))k$ can be represented by the set of integers of the form $o(\beta)k$, the result follows.  
\end{proof}
\subsection{ The Matching Graph}
The task of the matching module is to detect correct matches of $P_i$'s within $Y$. For this purpose we use a graph theoretic method. We define a graph
$G(V,E)$ with the vertex set as follows.
Graph $G$ has $k+1$ layers of vertices which are denoted
by $\Lambda_{0},\Lambda_{1},\cdots,\Lambda_{k}.$ {Each vertex in layer $\Lambda_i$, $1\leq i\leq k-1,$ represents a match of pivot $P_i$ in string $Y$.}
We refer to the vertices of $\Lambda_{i}$ and matches of $P_{i}$ in $Y$
interchangeably. For vertex $v\in \Lambda_i$, let $\check{v}$ and $\hat{v}$ denote, respectively, the first and the last indices of the match of $P_i$ corresponding to $v$ in $Y.$ We introduce two auxiliary vertices
$s$ and $t$ where $\Lambda_{0}=\left\{ s\right\} $ with $\hat{s}=0$
and $\Lambda_{k}=\left\{ t\right\} $ with $\check{t}=|Y|+1.$ Vertices
$s$ and $t$ represent the beginning and the ending of string $Y$ respectively.

We say a vertex in $\Lambda_i$ is a {\em good} vertex if it corresponds to a correct match of $P_i$ within $Y$. We call a vertex in $\Lambda_i$ a {\em bad}
vertex if it corresponds to an incorrect match of $P_i$. By definition of correct and incorrect matches, in each layer of graph $G$, there are possibly zero, one, or two good vertices. 
In order to detect the correct matches of $P_i$'s within $Y$, we need to find good vertices in graph $G$. For that, we  define the edge set of $G$ 
such that the good vertices are distinguished by their connectivity in the graph.

Let us define the distance between two vertices $u$ and $v$ in $G$ as follows:
\[\mbox{Dis}(u,v):=\check{v}-\hat{u}-1.\]
Notice that $\mbox{Dis}(u,v)$ is nonnegative only when the first bit of $v$ appears after the last bit of $u$. In that case, $\mbox{Dis}(u,v)$ is the number of bits between
$u$ and $v$ in $Y$.

For two pivots $P_i$ and $P_j$ with $i<j$ in $X$, the number of bits between them in $X$ is given by \[(j-i-1)L_P+(j-i)L_S.\]
If both $P_i$ and $P_j$ have correct matches in $Y$,  the number of bits  between the correct match for $P_i$ and the correct match for $P_j$ is at most $(j-i-1)L_P+(j-i)L_S$. 

Furthermore, in most cases, for $i<j$, the first bit of the correct match for $P_j$ appears after the last bit of the correct match for $P_i$. To see this, first notice that,
since the first bit of $P_j$ appears after the last bit of $P_i$ in $X$, if there are no deletions within $P_i$ and $P_j$, their order is preserved in $Y$. 

Now consider the following example: 
let $P_i=0000$ and $P_j=0000$. Also assume that all bits between $P_i$ and $P_j$ are deleted except for a single $0$ bit, and assume that exactly one bit is deleted from $P_i$ and exactly one bit is deleted from $P_j$.
In this case, the compound substring of $Y$ corresponding to $P_i$ and $P_j$ and the bits in between them in $X$ is $0000000$, where the first four bits constitute the correct match for $P_i$ and the last four bits constitute the correct match for $P_j$. As we can observe, the first
bit for the correct match of $P_j$ is the last bit for the correct match of $P_i$. The distance between the correct match for $P_i$ and the correct match for $P_j$  is $-1$. It is easy to verify that in general 
for $j>i$,  the least value of the distance between the correct match of $P_i$ and the correct match of $P_j$ is $-1$.    
 
Based on the two preceding observations, we connect a vertex $u\in \Lambda_i$ to a vertex $v\in \Lambda_j$ if and only if     
\begin{equation} \label{edge_condition}
-1 \leq \mbox{Dis}(u,v)\leq(j-i-1)L_P+(j-i)L_S.\end{equation}
Therefore, all pairs of good vertices from different layers are
connected together. By definition, $s$ and $t$, which indicate the beginning and the  ending of string $X$, respectively, are treated as ``auxiliary'' good vertices. Therefore, good vertices across different layers form an $s-t$ path in graph $G$. 
However, there are potentially many other pairs of vertices that satisfy the condition of \eqref{edge_condition} and are connected together. Figure \ref{graph} illustrates an instance of graph $G$ with $8$ layers and the connections between vertices. 
\begin{figure}
\centering
\includegraphics[scale=0.35]{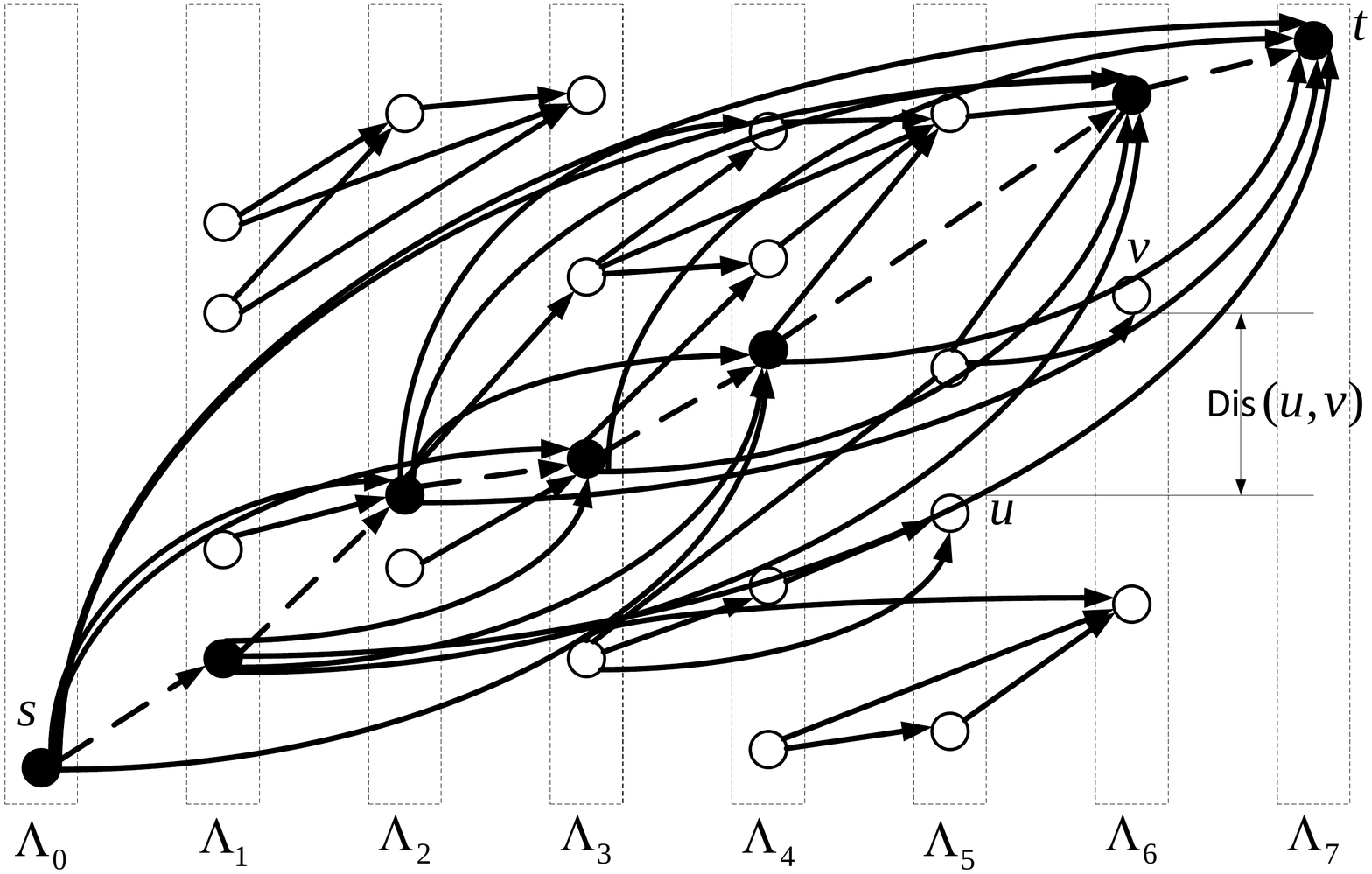}
\caption{\label{graph}Figure illustrates a graph $G$ with $8$ layers of vertices. The horizontal axis indicates different layers and the vertical axis indicates the position of each vertex in string $Y$ that can take values from $1$ to $|Y|$. The good and bad vertices are distinguished by black and white colors, respectively. The first layer has only one vertex $s$ and the last layer has only one vertex $t$. As it is seen, all good vertices in the graph are connected together and they form an $s-t$ path which is represented by the dashed edges in the graph.}
\end{figure}

The following 
theorem shows that, with very high probability,  bad vertices do not contribute to an $s-t$ path. That is, any $s-t$ path of the appropriate length in graph $G$ is formed mostly of good vertices. Recall the definition of $R$ from \eqref{R}. We then have the following result.  
\begin{thm}
\label{thm:For-a-random}Let  $X$ be a random input string to a deletion channel and $D$ be a random deletion pattern. Let $Y$ be the string obtained from $X$ and $D$. Let $G$ denote the matching graph corresponding to $Y$.
 Then, for $L_P\geq 11+2\log\frac{1}{\beta}$,
{with probability {at least} $1-2^{-\Omega(n)}$, all paths from $s$ to $t$ with $Rk+o(\beta)k$ vertices, have at least $Rk-\beta k+o(\beta)k$
good vertices.} 
\end{thm}  
Theorem \ref{thm:For-a-random} is not only an existence statement,
but also has an algorithmic implication. The implication is that
if we pick any path from $s$ to $t$ {with} $Rk+o(\beta)k$ vertices, the
path has many good vertices. Since finding an $s-t$ path of an appropriate length in $G$
is a computationally  tractable task (we will discuss the computational complexity in the next section), finding a large fraction of good
vertices is also a tractable task. 

{ \emph{Overview:}  Before presenting the detailed proof of the theorem we first sketch the overall idea of the proof. To prove the theorem, we show that for a random string $X$ and a random deletion pattern $D$, the probability of the existence of an $s-t$ path $Q$ in $G$ with $Rk+o(\beta)k$ vertices, such that the number of good vertices on $Q$ is less than $Rk-\beta k+o(\beta)k$  is upper bounded by $2^{-\Omega(n)}$. Equivalently, we  show that the probability of the existence of an $s-t$ path $Q$ with $Rk+o(\beta)k$ vertices such that number of bad vertices on $Q$ is more than $\beta k+o(\beta)k$  is upper bounded by $2^{-\Omega(n)}$.
To find an upper bound on the latter probability we use the union bound: for every $\alpha$ with $\beta \leq \alpha \leq R+o(\beta)$, we find an upper bound on the probability that  there exists an $s-t$ path $Q$ with $Rk+o(\beta)k$ vertices such that the number of bad vertices on $Q$ is  $\alpha k$. Then, by integrating the upper bound over all values of  $\alpha$ and showing that it is less than $2^{-\Omega(n)}$, we conclude the result. 

For a fixed value of $\alpha$ we evaluate an upper bound on the probability of the existence of an $s-t$ path $Q$ with $Rk+o(\beta)k$ vertices and $\alpha k$ bad vertices in the following way.
Let us denote all good vertices of $G$ by $U$, where $|U| = Rk+o(\beta)k$ with probability at least $1-2^{-\Omega(n)}$. We fix the realizations of all $\check{u}$ for which  $u\in U$. In other words, we fix the positions of good vertices of graph $G$. For $Q$ to have exactly $\alpha k$ bad vertices and $Rk-\alpha k + o(\beta)k$ good vertices, we first choose $Rk-\alpha k + o(\beta)k$ good vertices of $Q$ from the set $U$ (and account for the cases with possibly two correct matches). Graph $G$ has $k+1$ layers and  $Rk-\alpha k + o(\beta)k$ have been chosen to include the good vertices of $Q$. The remaining $\alpha k $ vertices of $Q$ are chosen from the remaining $k+1 - (Rk-\alpha k + o(\beta)k)$ layers. Since the vertices in set $U$ have fixed positions in $Y$, all good vertices of $Q$ have fixed positions in $Y$. However, we have only fixed the layers which include the bad vertices of $Q$ and not the positions of bad vertices in $Y$. Next, to find an upper bound on the number of possible  positions of the bad vertices of $Q$, we use a combinatorial argument based on the constraints imposed by the connectivity of consecutive vertices on $Q$ via the edges of graph $G$. We notice that  the positions  of all vertices on $Q$ are uniquely determined based on the distances between consecutive vertices on $Q$. Since the good vertices on $Q$ have fixed positions, the distance between two consecutive vertices on $Q$, where both of them are good vertices, is fixed. However, the distance of two consecutive vertices on $Q$, where one of the vertices is a bad vertex, is a variable. We need to find all solutions to these variables such that constraints defined by \eqref{edge_condition} are satisfied. We show how to consolidate all resultant edge constraints over all edges of $Q$ into a  single linear constraint, and then by counting the number of solutions to that constraint we find an upper bound on the number of possible positions of the bad vertices of $Q$. 

Finally, we notice that the probability that a substring $P_i$ has an incorrect match in $Y$ at some specific position is $2^{-L_P}$. Therefore, if we are given the positions of all bad vertices on $Q$, the probability that there are incorrect matches of the corresponding pivots at those positions is  $2^{-\alpha k L_P}$ each. By multiplying  $2^{-\alpha k L_P}$ by the upper bound on the possible number of positions for bad vertices of $Q$, we find an upper bound on the probability of the existence of $Q$ with $Rk+o(\beta) k$ vertices and $\alpha k$ bad vertices. Next we present the details of our argument.
}

\begin{proof}
We begin by finding an upper bound on the probability of the existence
of a path $Q$ from $s$ to $t$ with $Rk+o(\beta)k$ vertices out of which $\alpha k$
are bad vertices, for some $\beta\leq\alpha\leq1$.
There are $k+1$ layers in graph $G$ and by Lemma \ref{good_vertices}
with probability $1-2^{-\Omega(n)}$ there are $Rk+o(\beta)k$ {layers with} good vertices
in graph $G$. Let us fix the realization of the deletion pattern $D$, the realization of the pivots $P_i$ in $X$ with exactly one deletion, and the realization of the immediate undeleted bits before and after pivots $P_i$ in $X$ with exactly one deletion. 
In this way, good vertices of graph $G$ are fixed.
We consider two cases:
\subsection*{Case 1: $\beta \leq \alpha<\frac{1}{2}$}
For  $\beta \leq \alpha<\frac{1}{2}$, first we fix the layers which have a vertex on the path
$Q$ of length $Rk+o(\beta)k$. Since by assumption, there are $Rk-\alpha k+o(\beta)$ good vertices on the path $Q$,
the selection of good and bad vertices on the path can be done in at most the following number of ways 
 \begin{align}
\binom{Rk+o(\beta)k}{Rk-\alpha k+o(\beta)k}\cdot 2^{o(\beta)k} \cdot\binom{(1-R)k+\alpha k+o(\beta)k}{\alpha k}& <\nonumber\\
2^{k((R+o(\beta))H(\frac{\alpha}{R+o(\beta)})+o(\beta)+(1-R+\alpha+o(\beta))H(\frac{\alpha}{1-R+\alpha+o(\beta)}))} & =\nonumber\\
2^{k(RH(\frac{\alpha}{R})+(1-R+\alpha)H(\frac{\alpha}{1-R+\alpha})+o(\beta))}.\label{delta_1}\end{align}
In the multiplication above, the first term stands for the number of ways we can choose the layers
with good vertices on the path $Q$. By Lemma \ref{two_matches}, with probability $1-2^{-\Omega(n)}$, there are  $o(\beta)k$ pivots with two correct matches in $Y$. Thus, with probability $1-2^{-\Omega(n)}$, there are at most  $o(\beta)k$ layers with two good vertices among the layers with good vertices on $Q$. Therefore, the second term in the multiplication above, is an upper bound on the  number of combinations we can pick one good vertex from each of those layers. The last term stands for the number of ways
we can choose the layers with bad vertices from the remaining available layers.
Notice that the layers with bad vertices can be chosen from all $k+1$ layers except the $Rk-\alpha k +o(\beta) k$ layers which are chosen to have good vertices. 
  Also, the inequality holds by application of the inequality $\binom{N}{\varepsilon N} < 2^{H(\varepsilon) N}$ for $0< \varepsilon<\frac{1}{2}$  and positive integer  $N$, where $H(\cdot)$ is the binary entropy function.

Suppose that path $Q$ has vertices from layers $\Lambda_{i_{1}},\Lambda_{i_{2}},\cdots,\Lambda_{i_{Rk+o(\beta)k}}$.
Let $\mathcal{I}:=\left\{ 1,\cdots,Rk+o(\beta)k\right\}$ be the set of indices of the layers with a vertex on the path $Q$. Let $\mathcal{I}=\mathcal{I}_{g}\cup\mathcal{I}_{b}$
where $\mathcal{I}_{g}$ is the set of indices of layers with good
vertices on $Q$ and $\mathcal{I}_{b}$ is the set of indices of layers with
bad vertices on $Q$ (The sets $\mathcal{I}_{g}$ and $\mathcal{I}_{b}$ are disjoint.).
 That is,  layer $\Lambda_{i_{j}}$ with $j\in\mathcal{I}_{g}$
is a layer with a good vertex on $Q$ and layer $\Lambda_{i_{j}}$ with
$j\in\mathcal{I}_{b}$ is a layer with a bad vertex on $Q$.

Let us express the {path} $Q$ as $s-v_{i_{1}}-v_{i_{2}}-\cdots-v_{i_{Rk+o(\beta)k}}-t$
where $v_{i_{j}}\in\Lambda_{i_{j}}.$ 
Path $Q$ is uniquely identified by the position of the first bit of  its vertices, $\left(\check{v}_{i_{1}},\check{v}_{i_{2}},\cdots,\check{v}_{i_{Rk+o(\beta)k}}\right)$.
Equivalently, if we know the distance between consecutive vertices $\left(\mbox{Dis}(v_{i_j},v_{i_{j+1}}):j\in {\mathcal I} \right)$, we can 
uniquely identify the position of each vertex on the path. Therefore, next we count the 
number of possible values of the distances between consecutive vertices $\left(\mbox{Dis}(v_{i_j},v_{i_{j+1}}):j\in {\mathcal I} \right)$.

Since good vertices are pinned down on the path, the value of $\mbox{Dis}(v_{i_j},v_{i_{j+1}})$ is determined if both $v_{i_j}$ and $v_{i_{j+1}}$
are good vertices. Let us define set ${\mathcal H} \subset {\mathcal I}$ as follows
\[{\mathcal H}=\{j : j \in {\mathcal I}_b \vee (j+1) \in {\mathcal I}_b \}.\]
Therefore $\left(\mbox{Dis}(v_{i_j},v_{i_{j+1}}):j\in {\mathcal H} \right)$ is the set of distances between consecutive vertices of $Q$ that are undetermined. 
The number of bad vertices on $Q$ is $\alpha k$. Therefore   $|{\mathcal H}| \leq 2\alpha k$. 
Let $j_{1}$ and $j_{2}$ with $j_{1}<j_{2}$ be two consecutive elements in the ordered version of $\mathcal{I}_{g}.$   
Then by additivity of distances, bad vertices $v_{i_{j_{1}+1}},\cdots,v_{i_{j_{2}-1}}$
need to satisfy the following constraint: 
 \begin{equation}
\label{eq:main_1}
\sum_{t=j_{1}}^{j_{2}-1}\mbox{Dis}(v_{i_{t}},v_{i_{t+1}})=\mbox{Dis}(v_{i_{j_{1}}},v_{i_{j_{2}}})-(j_{2}-j_{1}-1)L_{P},
\end{equation}
where $(j_{2}-j_{1}-1)L_{P}$ is the total length of the substrings $v_{i_{j_{1}+1}},\cdots,v_{i_{j_{2}-1}}$ in $Y$.  
Furthermore, bad vertices should be placed on $Q$ such that they satisfy the  constraint given in (\ref{edge_condition}). For every $j\in {\mathcal H}$, we need to have 
\begin{equation}
-1\leq \mbox{Dis}(v_{i_j},v_{i_{j+1}})\leq(i_{j+1}-i_{j})L_{S}+(i_{j+1}-i_{j}-1)L_{P}.
\label{eq:main_2}
\end{equation}
Next we find an upper bound on the number of integer vectors $\left(\mbox{Dis}(v_{i_j},v_{i_{j+1}}):j\in {\mathcal H} \right)$
that satisfy (\ref{eq:main_1}) and (\ref{eq:main_2}).

For $j\in \mathcal I$, we use the following 
change of variables 
\[\delta_{j}:=(i_{j+1}-i_{j})L_{S}+(i_{j+1}-i_{j}-1)L_{P}-\mbox{Dis}(v_{i_j},v_{i_{j+1}}).\]
Equation (\ref{eq:main_1}) in terms of the variables $\delta_j$'s is written as follows. For $j_1<j_2$, as any two consecutive
elements in the ordered version of ${\mathcal I}_g$, we have 
\begin{equation}
\label{eq:main_3}
\sum_{j=j_{1}}^{j_{2}-1}\delta_{j}=(i_{j_{2}}-i_{j_{1}})L_{S}+(i_{j_{2}}-i_{j_{1}}-1)L_{P}-\mbox{Dis}(v_{i_{j_{1}}},v_{i_{j_{2}}}).
\end{equation}
Observe that in  (\ref{eq:main_3}),  $(i_{j_{2}}-i_{j_{1}})L_{S}+(i_{j_{2}}-i_{j_{1}}-1)L_{P}$ is the 
number of bits between $P_{i_{j_1}}$ and $P_{i_{j_2}}$ in $X$ and $\mbox{Dis}(v_{i_{j_{1}}},v_{i_{j_{2}}})$ is the number of 
bits between the correct match of $P_{i_{j_1}}$ and the correct match of $P_{i_{j_2}}$ in $Y$. Therefore, the right hand side of 
Equation (\ref{eq:main_3}) is the number of deleted bits in the substring between $P_{i_{j_1}}$ and $P_{i_{j_2}}$ in $X$.
To find an upper bound on the number of solutions  for $(\delta_j: j\in {\mathcal H}),$ we relax 
constraints in (\ref{eq:main_3}) over all $j$'s into a single constraint by adding them together:
\begin{equation}
\sum_{j\in {\mathcal H}} \delta_j = \delta -\sum_{j'\in\mathcal{H}^{c}}\delta_{j'}. 
\end{equation}
Here $\delta$ is the total number of deleted bits from $X$ and set
$\mathcal{H}^{c}=\mathcal{I}\setminus\mathcal{H}$ is the set
of indices $j'$ for which $h_{j'}$ is determined; i.e., $v_{i_{j'}}$
and $v_{i_{j'+1}}$ are both good vertices. Furthermore, $\delta_{j'}$
is the number of deleted bits from the substring between  $P_{i_{j'}}$ and $P_{i_{j'+1}}$ in $X$.

Next, we use the following result on the concentration of $\sum_{j\in {\mathcal H}} \delta_j $
around its expected value. 
\begin{lm}\label{concentration}
For a random string $X$, random deletion pattern $D$, and the resultant string $Y$, the following bound holds:
 \[
\Pr\left\{\left|\sum_{j\in\mathcal{H}}\delta_{j}-\mathbb{E}\left[\sum_{j\in\mathcal{H}}\delta_{j}\right]\right|=o(\beta)k\right\}\geq1-2^{-\Omega(n)}.\]
\end{lm}
\begin{proof}
See Appendix II.
\end{proof}
To estimate $\mathbb{E}\left[\sum_{j\in\mathcal{H}}\delta_{j}\right]$, first notice that the average number of deleted bits from $X$ is ${\mathbb E}\left[\delta\right]=n\beta=(1+o(\beta))k$.

Next we find $\mathbb{E}\left[\delta_{j'}\right]$ for $j'\in\mathcal{H}^{c}.$
Since $Q$ has $Rk+o(\beta)k$ vertices, the average size of the substring
between $P_{i_{j'}}$ and $P_{i_{j'+1}}$ in $X$ is $\frac{n}{Rk+o(\beta)k}.$
Therefore, $\mathbb{E}\left[\delta_{j'}\right]$, the average number of deleted
bits from the substring between $P_{i_{j'}}$ and $P_{i_{j'+1}}$
in $X$ is \[
\mathbb{E}\left[\delta_{j'}\right]=\frac{n\beta}{Rk+o(\beta)k}=\frac{1}{R+o(\beta)}=1+\beta L_{P}-2\beta+o(\beta).\]
 Since $|\mathcal{H}|\leq2\alpha k$ and $|\mathcal{I}|=|\mathcal{H}|+|\mathcal{H}^{c}|=(R+o(\beta))k$,
we find that \[
|\mathcal{H}^{c}|\geq(R-2\alpha+o(\beta))k=(1-\beta L_{P}+2\beta-2\alpha+o(\beta))k.\]
 We conclude that 
\begin{align*}
 \mathbb{E}[\sum_{j\in\mathcal{H}}\delta_{j}] & = \mathbb{E}[\delta] - \mathbb{E}[\sum_{j' \in {\mathcal H}^c}\delta_{j'}]\\
 & =k-|\mathcal{H}^{c}|\mathbb{E}\left[\delta_{j'}\right]+o(\beta)k\\
 & \leq k\cdot\\
& (1-(1-\beta L_{P}+2\beta-2\alpha)(1+\beta L_{P}-2\beta)+o(\beta))\\
 & =2\alpha k(1+\beta L_{P}-2\beta)+o(\beta)k,\end{align*}
and therefore by Lemma \ref{concentration}, with probability at least $1-2^{-\Omega(n)}$
\begin{equation}
\label{sum_delta}
\sum_{j\in\mathcal{H}}\delta_{j}=2\alpha k(1+\beta L_{P}-2\beta)+o(\beta)k.
\end{equation}
Therefore, we showed that  (\ref{eq:main_1}) yields the weaker constraint in $\eqref{sum_delta}$ on the vector $(\delta_j: j\in {\mathcal H})$.  

Now consider the inequality in (\ref{eq:main_2}). We can rewrite it in terms of $\delta_j$ as follows
\[
0\leq\delta_{j}\leq(i_{j+1}-i_{j})L_{S}+(i_{j+1}-i_{j}-1)L_{P}+1.\]
 To find an upper bound on the number of solutions for $(\delta_j: j\in {\mathcal H})$, we relax the preceding constraint
to $\delta_{j}\geq 0.$ 

 Under the constraint that $\delta_{j}\geq0,$ the number of integer solutions
for $(\delta_{j}:j\in {\mathcal H})$ under the condition \eqref{sum_delta}, is given by
\begin{align}
\binom{2\alpha k(1+\beta L_{P}-2\beta)+o(\beta)k+|\mathcal{H}|-1}{|\mathcal{H}|-1} & \leq \nonumber\\
\binom{2\alpha k(2+\beta L_{P}-2\beta+\frac{o(\beta)}{\alpha})}{2\alpha k} & \leq \nonumber\\
2^{2\alpha k(2+\beta L_{P}-2\beta+\frac{o(\beta)}{\alpha})} & \leq2^{5\alpha k},\label{delta_2}
\end{align}
 where the last estimate holds for sufficiently small $\beta.$ 

Given the number of possibilities for path $Q$, we next compute the probability of occurrence of each realization of path $Q$.
Since $X$ is generated by an i.i.d. Bernoulli source of parameter $\frac{1}{2}$ and deletions occur independently, $Y$ is also generated by an i.i.d. Bernoulli source of parameter $\frac{1}{2}$ and  different substrings of $Y$ are independent. Therefore, the probability of  any given realization of  bad vertices  as specified
by the choice of $\delta_{j}$'s is $2^{-L_{P}\alpha k}$. By applying the union bound on the probability of existence of individual paths, and using inequalities \eqref{delta_1} and \eqref{delta_2}, we conclude that the probability
of the existence of a path $Q$ with $Rk+o(\beta)k$ total vertices and $\alpha k$ bad vertices  is upper bounded
by $2^{\Delta_{\alpha}k}$ where
\begin{align*}
\Delta_{\alpha} & =RH\left(\frac{\alpha}{R}\right)+(1-R+\alpha)H\left(\frac{\alpha}{1-R+\alpha}\right)\\
 & \quad+5\alpha- L_{P}\alpha+o(\beta)\\
 & =-\alpha\log\alpha+\alpha\log R-R\log(1-\frac{\alpha}{R})\\
 & \quad+\alpha\log(1-\frac{\alpha}{R})-\alpha\log\alpha+\alpha\log(1-R+\alpha)\\
 & \quad-(1-R)\log(1-\frac{\alpha}{1-R+\alpha})+5\alpha-L_{P}\alpha +o(\beta).\end{align*}
 Next we find an upper bound for $\Delta_{\alpha}.$ 
Since for sufficiently small $\beta$, $R<1,$ then $\alpha\log R<0.$
Since $\alpha < \frac{1}{2}$,
for small enough $\beta$, $R>\alpha$. 
Therefore $\alpha\log(1-\frac{\alpha}{R})<0$ and $\alpha\log(1-R+\alpha)<0.$
Using the inequality $\log(1+x)\leq\frac{x}{\ln2}$ for $x>-1$ we
find that 
\begin{align*}
-R\log(1-\frac{\alpha}{R})=R\log(1+\frac{\alpha}{R-\alpha}) &\leq\frac{R\alpha}{(R-\alpha)\ln2} \\
                                                                                                 &= \frac{\alpha}{(1-\frac{\alpha}{R})\ln2}.
\end{align*}
Notice that for small values of $\beta,$ $R$ is close to $1$. Let us assume that $\beta$ is sufficiently small such that $R> 0.9$.
Since $\alpha < \frac{1}{2}$, we can write
\[-R\log(1-\frac{\alpha}{R}) \leq \frac{\alpha}{(1-\frac{\alpha}{R})\ln2} < \frac{\alpha}{(1-\frac{1}{1.8})\ln 2} = \frac{2.25 \alpha}{\ln 2}.\]
Also we have 
\begin{align*}
-(1-R)\log\left(1-\frac{\alpha}{1-R+\alpha}\right) & =(1-R)\log\left(1+\frac{\alpha}{1-R}\right)\\
 & \leq\frac{(1-R)\alpha}{(1-R)\ln2}=\frac{\alpha}{\ln2}.\end{align*}
  Therefore \begin{align*}
\Delta_{\alpha} & \leq o(\beta)-2\alpha\log\alpha+\frac{2.25\alpha}{\ln2}+\frac{\alpha}{\ln2}+5\alpha- L_{P}\alpha\\
 & =\alpha(\frac{o(\beta)}{\alpha}-2\log\alpha+\frac{3.25}{\ln2}+5-L_{P})\\
& < \alpha(\frac{o(\beta)}{\alpha}-2\log\alpha+9.7-L_{P}).
\end{align*}
Notice that $\frac{o(\beta)}{\alpha}\leq\frac{o(\beta)}{\beta}$  is arbitrarily small for sufficiently small 
 $\beta$. If we choose $\beta$ such that $\frac{o(\beta)}{\alpha} <0.3$, then we have
\[\Delta_{\alpha} <  \alpha( -2\log\alpha+10 - L_{P}).\]  
\subsection*{Case 2: $\frac{1}{2} \leq \alpha \leq R+o(\beta)$}
We again seek to bound the probability of the existence of a path $Q$ from $s$
to $t$ with $Rk+o(\beta)k$ total vertices and $\alpha k$ bad vertices.
Let the path $Q$ be denoted by $s-v_{i_{1}}-v_{i_{2}}-\cdots-v_{i_{Rk+o(\beta)k}}-t$
and let $\delta_{j}$ denote the number of deleted bits between {vertices}
$v_{i_{j}}$ and $v_{i_{j+1}}.$ Clearly, the sum of $\delta_{j}$ is
the total number of deletions in string $Y$. By Theorem \ref{Hoeffding}, with probability at least $1-2^{-\Omega(n)}$ we have 
$\sum_{j=0}^{Rk+o(\beta)k}\delta_{j} = n\beta+ n\beta o(\beta) = k(1+o(\beta))$.
The number of integer solutions for $\delta_{j}\geq0$ under {this} constraint is
\[
\binom{k+Rk+o(\beta)k-1}{Rk+o(\beta)k-1}\leq\binom{2k}{k}\leq 2^{2k}.\]
 The probability for each solution of $\delta_{j}$'s to represent
a valid $s-t$ path is at most $2^{-L_{P}\alpha k}$. Therefore, an
upper bound on the probability of existence of  a path $Q$ in this
case is $2^{(2-L_{P}\alpha)k}.$

Finally, putting both cases for the range of $\alpha$ together, the probability of the existence of a path $Q$ with $Rk+o(\beta)k$
vertices between $s$ and $t$ with at least $\beta k$ bad vertices can
be upper bounded by the sum of two integrals:
\begin{align*}
 & \int_{\alpha=\beta}^{\frac{1}{2}}2^{\Delta_{\alpha}k}d\alpha k+\int_{\alpha=\frac{1}{2}}^{R+o(\beta)}2^{(2-L_{P}\alpha)k}d\alpha k\leq\\
 & \int_{\alpha=\beta}^{\frac{1}{2}}2^{(-2\log\alpha+10-L_{P})\alpha k}d\alpha k+\int_{\alpha=\frac{1}{2}}^{R+o(\beta)}2^{(2-L_{P}\alpha)k}d\alpha k.\end{align*}
 If we pick 
$L_{P}\geq 11+2\log\frac{1}{\beta}$ 
then we find  \[(-2\log\alpha+10-L_{P})\alpha k \leq -\alpha k\leq-\beta k\] 
for $\beta\leq \alpha\leq\frac{1}{2}.$ Also, 
\[2-L_P \alpha  \leq 2-\frac{1}{2} \cdot 11 =-3.5. \]
Therefore, we can upper bound the sum of the two integrals by
\begin{align*}
\int_{\beta}^{\frac{1}{2}}k2^{-\beta k}d\alpha+\int_{\frac{1}{2}}^{R+o(\beta)}k2^{-3.5k}d\alpha\leq\frac{k}{2}(2^{-\beta k}+2^{-3.5k})\\
=2^{-\Omega(n)}.\end{align*}
This completes the proof of Theorem \ref{thm:For-a-random}.
\end{proof}
\begin{figure*}
\centering
\includegraphics[scale=0.35]{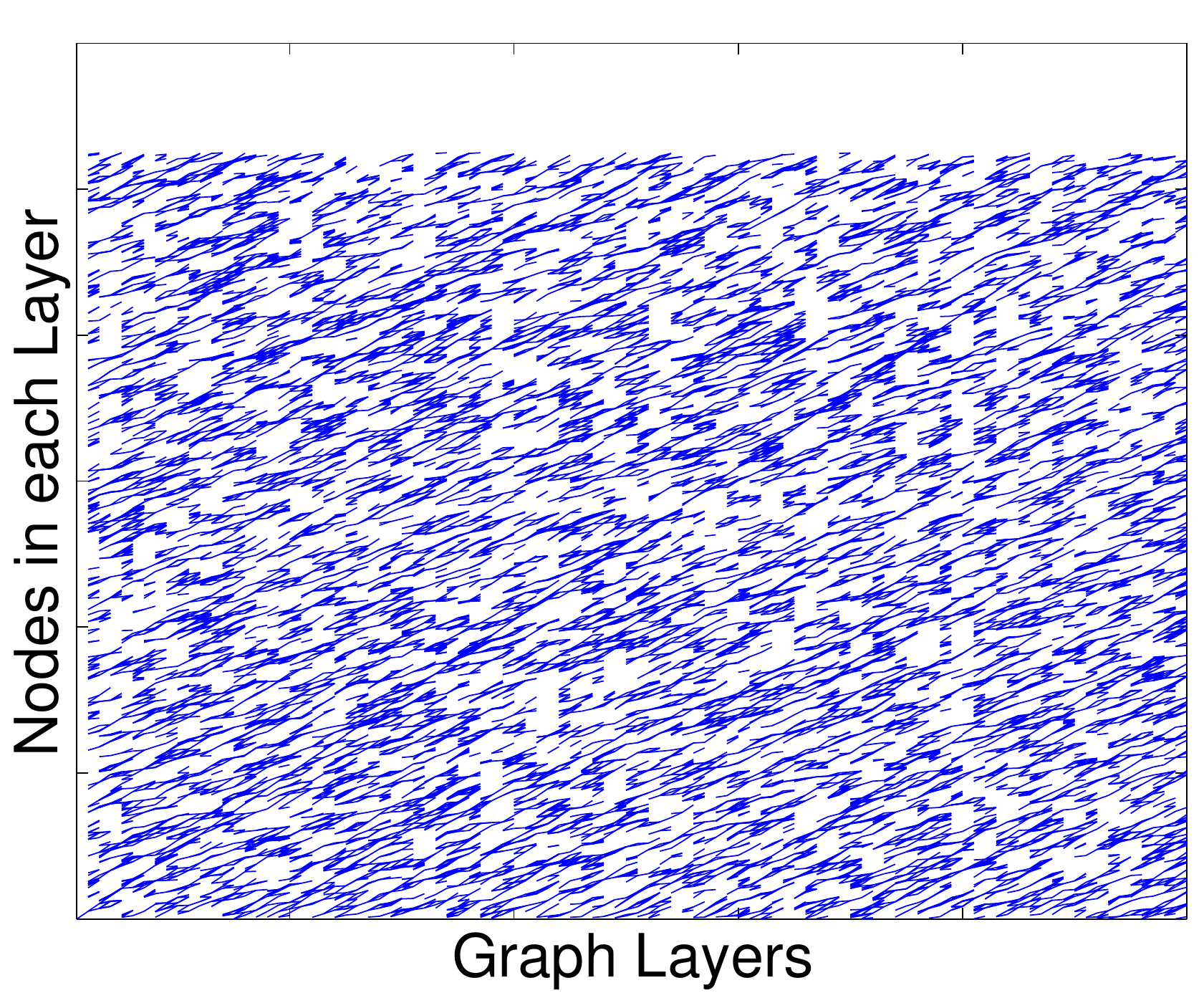}
\includegraphics[scale=0.35]{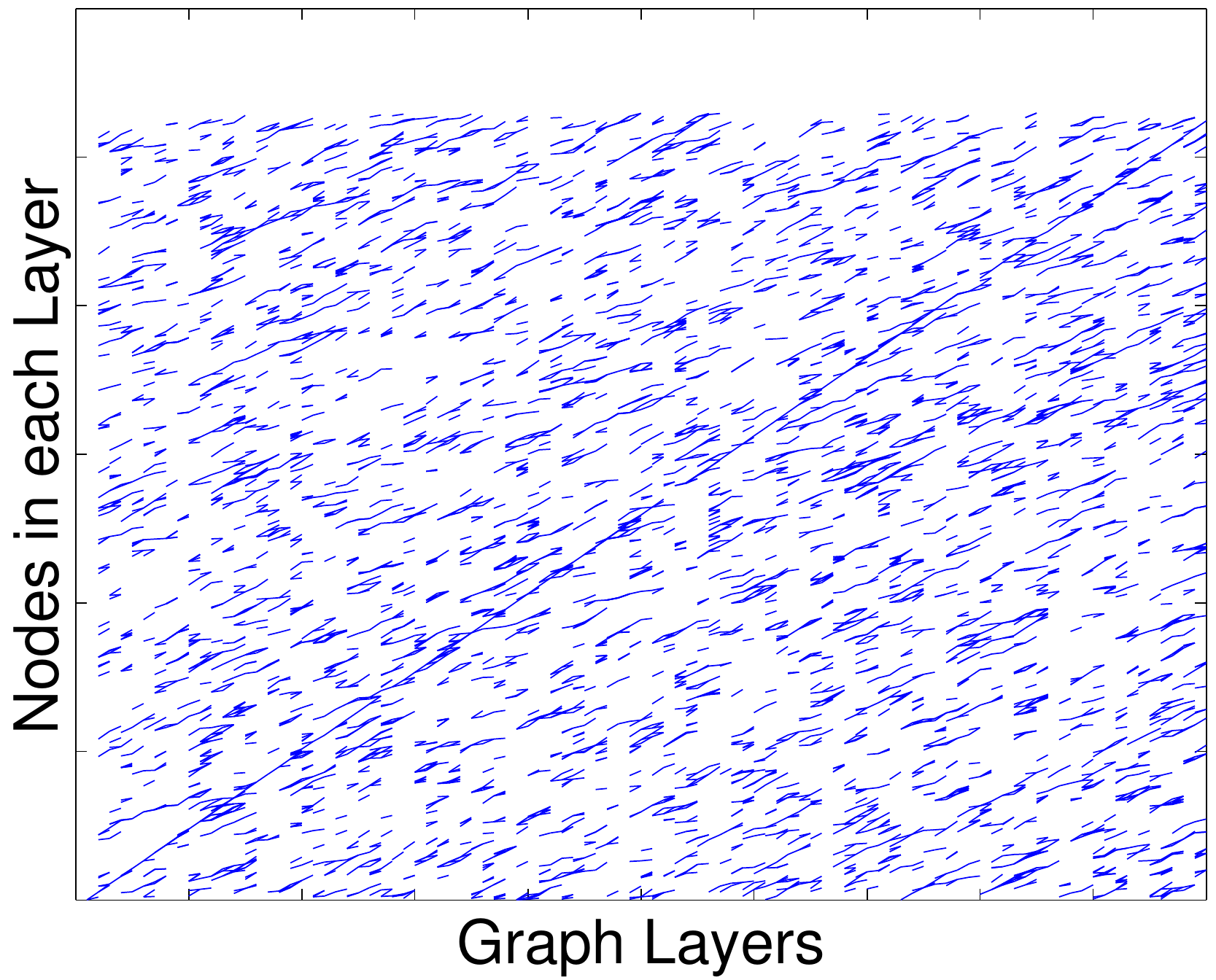}
\includegraphics[scale=0.35]{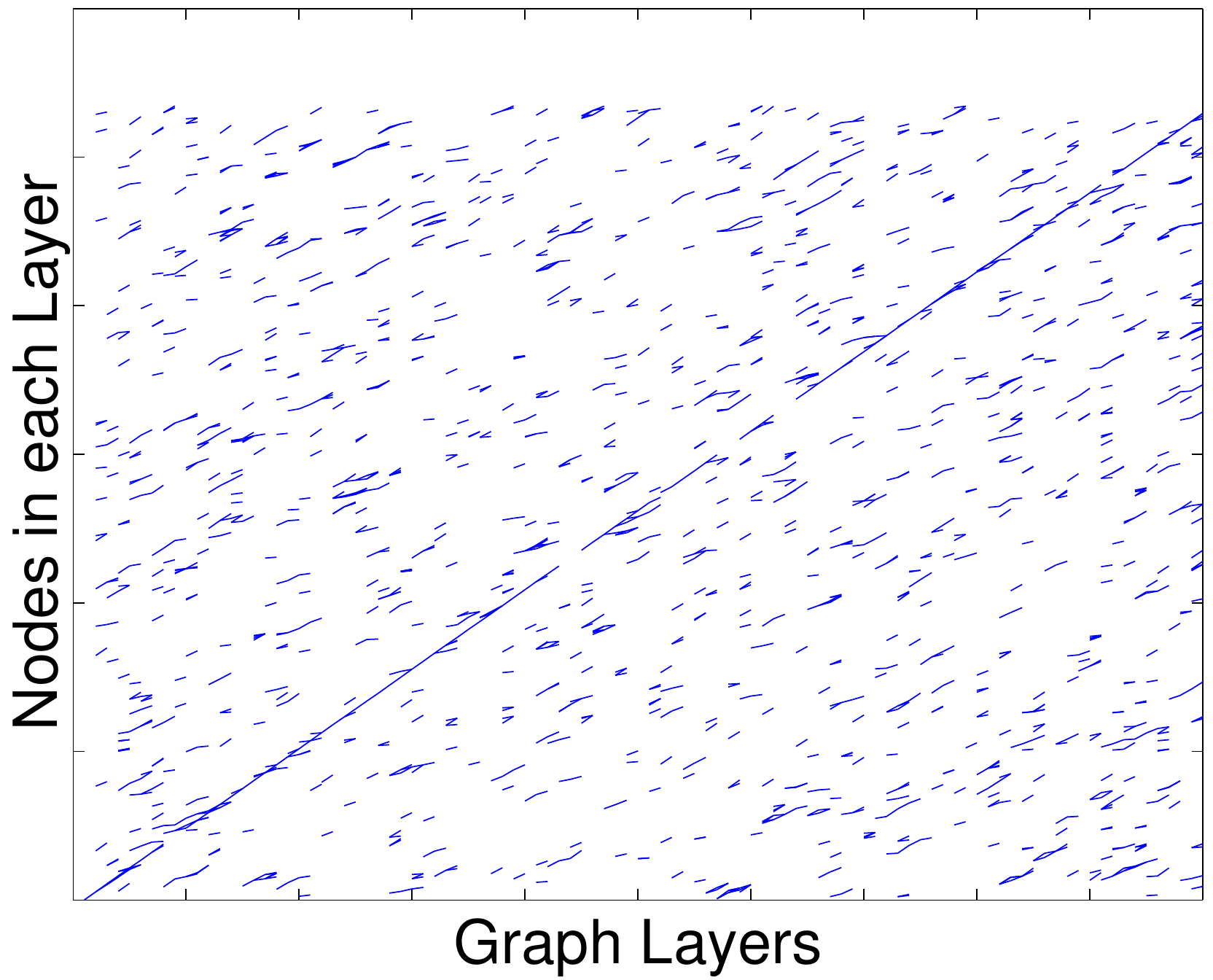}
\caption{\label{fig:Graph--for}Graph $G$ for $k=100,\beta=0.01,L_{S}=100,$
and $L_{P}=6,7$, and $8$, where only the edges between consecutive layers
are depicted.}
\end{figure*}
In order to verify the result of Theorem \ref{thm:For-a-random} in
a practical setting, we have plotted graph $G$ for a randomly generated string
$X$ and a randomly generated deletion pattern $D$ with parameter $\beta=0.01,$
for {three} values of $L_{P}$ in Figure \ref{fig:Graph--for}. To avoid
visual complications, we have only plotted edges that connect vertices
on two consecutive layers. As it is clear from the figure, for small
values of $L_{P},$ there are many edges in the graph and there are
potentially many paths that connect $s$ to $t$ which do not share
many vertices with the correct path. However, for larger values of $L_{P},$
the irrelevant edges disappear from the graph and the only path that
remains is the one formed by good vertices of the graph. For $\beta=0.01$,
Theorem \ref{thm:For-a-random} states that $L_{P}\geq 11+2\log\frac{1}{\beta}\approx17$
is sufficient for our purpose. In practice, we observe values
of $L_{P}$ around $8$ are sufficient for distinguishing good vertices
on graph $G.$
\section{\label{sec:Practical-Algorithms}Practical Implementation}
In this section we discuss practical implementation of our synchronization
protocol, consisting of a matching module, a deletion recovery module, and an LDPC decoder module (see Figure \ref{scheme}).  
For the deletion recovery module, we can implement the synchronization protocol of Venkataramanan {\em et al.} \cite{Venkataramanan_Zhang_Ramchandran} 
which runs in linear time in $|F_j|$ for deletion recovery of each substring $F_j,1\leq  j\leq k'$.  Therefore the overall complexity of the deletion recovery module is linear in $n$. 
For the LDPC decoder module there are many sophisticated encoding and decoding schemes (see \cite{Richardson_Urbanke_encoding,Richardson_Urbanke}) that need running  time
linear  in $n$.

In this section we therefore focus on the implementation
of the graph-based algorithm for
the matching  module explained in the previous section.
The result of Theorem \ref{thm:For-a-random} indicates that to find
a large number of correct matches for pivots in the received string $Y$,
it suffices to find an $s-t$ path with $Rk+o(\beta)k$ vertices in the  matching  {graph}
$G.$ We now argue that this  problem can be cast as the well
known {}``shortest path problem'' in a directed graph, so it can be
efficiently solved in polynomial time.

As the first step, we only keep the vertices in {graph} $G$ which have an edge
to vertex $t$ and remove all other vertices. Since all good
vertices are connected to vertex $t$, this step does not eliminate any good
vertex from {graph} $G$. Let $\tilde G$ denote the resulting graph. As the second step, we find the longest $s-t$ path in
$\tilde G$. Since all good vertices are connected together and form an $s-t$ path of length $Rk+o(\beta)k$, the longest path in $\tilde G$  has at least $Rk+o(\beta)k$ vertices. Finally,
we modify the discovered path into a path with only $Rk+o(\beta)k$
vertices by keeping  only the first $Rk+o(\beta)k$
vertices on the path. Since each vertex in  graph $\tilde G$ has an edge to {vertex} $t$,
the resulting vertices from this step form a path with $Rk+o(\beta)k$ vertices
from $s$ to $t$.

The only step of the above procedure which is computationally demanding
is the second step for finding the longest $s-t$ path in $\tilde G$. Notice that since $G$ and hence $\tilde G$ are acyclic graphs,  the longest $s-t$ path problem in $\tilde G$, can
be reduced to the shortest $s-t$ path problem in $\tilde G$ by assigning weight $-1$ to each edge. The latter problem is solvable in time $O(|{\tilde G}|^{2})$,
for instance by Dijkstra's algorithm \cite{Dijkstra}, where $|{\tilde G}|$
is the number of vertices in ${\tilde G}.$ We upper bound $|{\tilde G}|$ by $|G|$. To approximate $|G|,$ we notice
that there are  $n\beta + o(1)$ layers in {graph} $G$ and the number of vertices
in {layer} $\Lambda_{i}$ is the number of copies of pivot $P_{i}$
in $Y$, which is on average $2^{-L_{P}}|Y|=O({\beta^2}n).$ Therefore, $ |G| = O((\beta^2 n)\cdot(n\beta) + o(1))=O(n^2\beta^3).$ We conclude that
the complexity of matching pivots in graph $G$ is upper bounded by $O(|G|^{2})=O(n^4\beta^6).$
\section{Conclusions} \label{conclusion}
In this paper we offered the first synchronization protocol for recovering from a small rate of deletions
with an optimal order of transmitted bits and with exponentially small reconstruction error.
The main idea was to divide the synchronization problem into synchronization between shorter substrings of the source file and the destination file.
For that, our protocol sends equally spaced small substrings of the source file to the destination, and destination then uses a graph theoretic
algorithm to locate the short substrings within its file with high accuracy. For synchronization between the shorter substrings we used existing 
protocols that recover from a small number of edits. We observed that the compound output of the first two steps can be modeled as an output of a BSC with a
small error probability. This error can be recovered with a low bit error rate by using an LDPC coding scheme.

While in this work we only considered recovering from i.i.d. patterns of deleted bits, there are many other interesting edit models that the ideas of this paper can be applied to. 
An immediate extension of our work is to the synchronization from i.i.d. insertions. To explain an i.i.d. insertion process, let us consider an equivalent description of the i.i.d. deletion process
considered in this paper. In the new description, the deletion pattern $D$ is described as an independent sequence of positive integers, where the integers alternatively represent the length 
of  zero and one runs in the deletion pattern $D$. It is easy to verify that if the integers are generated independently according to an appropriate geometric distribution, the result 
is an i.i.d. 0-1 deletion pattern. We can describe the insertion pattern in the same way by generating the run length sequence of the pattern. For the insertion pattern, each run of ones corresponds to 
an inserted substring of equal length generated by an i.i.d. Bernoulli process. Also, each run of zeros corresponds to a substring of the input string of equal length in the output. It is not hard to see that the solution presented in this paper for synchronization from deletions is directly  applicable to solving the synchronization problem from random insertions.
One can also consider more general patterns of deletions or insertions, e.g., the 0-1 deletion (insertion) patterns that follow a Markov chain random process (see \cite{Ma_Ramachandran_Tse}).

Another interesting direction for the extension of  this work is the design of synchronization protocols that are capable of recovering from a small rate of both deletions and insertions. While the deletion recovery module 
in our work, based on the algorithm by Venkataramanan {\em et al.} \cite{Venkataramanan_Zhang_Ramchandran}, is directly applicable to recovery from deletions and insertions, the main challenge is to extend 
the graph theoretic algorithm for matching the pivot substrings in the received string $Y$ when there are both deletions and insertions. Again, many parts of our argument still hold for the new setting as long as 
the edits happen with small rates while some technical parts may need to be modified. This extension is the focus of our current research. Our recent progress is reported in \cite{Lara_Behzad}.

There are some other aspects of our current research that can be modified into more efficient synchronization protocols. For example, our algorithm needs a small backward bandwidth from node $B$ to node $A$ in the deletion recovery module. This bandwidth is an inherent component of the synchronization protocol  of Venkataramanan {\em et al.}, \cite{Venkataramanan_Zhang_Ramchandran}. It is of great interest to design protocols that 
can operate on forward links only. As proved by Orlitsky \cite{Orlitsky}, design of optimal protocols for recovery from deletions on forward links implies optimal protocols for recovery from deletions and insertions. Furthermore, such protocols can be implemented as efficient channel codes for communicating over edit channels (see \cite{Ma_Ramachandran_Tse, Kanoria_Montanari, Kalai_Mitzenmacher_Sudan, Mitzenmacher}).

Finally, from a practical perspective, it is interesting to design a more efficient implementation of the graph theoretic matching algorithm which is at the heart of our matching module. While our algorithm runs in $O(n^4\beta^6)$ time, we believe that by exploiting the specific structure of the matching graph, and applying additional restrictions on the connectivity of the vertices of the graph together, it is possible to considerably reduce the running time of the matching module and hence reduce the overall complexity of the synchronization protocol.     

\section*{Acknowledgement}
The work is supported in part by NSF CAREER grant no. CCF-1150212, gift from Intel Corporation, Okawa Research Grant and Intel Early Career Award. L. Dolecek acknowledges helpful discussions with Nicolas Bitouze.
\bibliographystyle{IEEEtran}
\bibliography{mmor_refrence_v3}
\section*{Appendix I}
Here we evaluate $\mathbb{E}\left[\delta_{j}\log|F_{j}|\right]$.
\begin{align*}
\mathbb{E}\left[\delta_{j}\log|F_{j}|\right] & =\sum_{l}\Pr\left\{|F_{j}|=l\right\}\mathbb{E}\left[\delta_{j}\log|F_{j}|\bigl||F_{j}|=l\right]\\
 & =\sum_{l}\beta l\log l\Pr\left\{|F_{j}|=l\right\}\\
 & =\mathbb{E}\left[\beta|F_{j}|\log|F_{j}|\right].\end{align*}
 Next we estimate $\mathbb{E}\left[\beta|F_{j}|\log|F_{j}|\right].$ Recall that
$F_{j}$ is the substring of $X$ between $P_{i_{j-1}}$ and $P_{i_{j}}.$
There are $(i_{j}-i_{j-1})$ segment strings and $(i_{j}-i_{j-1}-1)$
pivot strings between $P_{i_{j-1}}$ and $P_{i_{j}}.$ Therefore \[
|F_{j}|=(i_{j}-i_{j-1})L_{S}+(i_{j}-i_{j-1}-1)L_{P}.\]
 There is a total of $k$ pivots, and $k'$ of them are matched by the matching module. Therefore, with probability $p:=\frac{k'}{k}=(1-L_{P}\beta+2\beta+o(\beta))$  pivot $P_i$ is matched.
Furthermore, the probability that
a pivot is matched is independent of other pivots. Thus, $(i_{j}-i_{j-1})$
has the following geometric distribution\[
\Pr\left\{ i_{j}-i_{j-1}=\ell \right\} =p(1-p)^{\ell-1}.\]
Suppose $r\in\left\{ 1,2,\cdots\right\} $ is a random variable distributed
as above. If we upper bound $|F_{j}|\leq(i_{j}-i_{j-1})(L_{S}+L_{P})$,
then \begin{align*}
\mathbb{E}\left[\beta|F_{j}|\log|F_{j}|\right] & \leq\mathbb{E}\left[\beta r(L_{S}+L_{P})\log r(L_{S}+L_{P})\right]\\
 & =\beta(L_{S}+L_{P})\mathbb{E}\left[r\log r+r\log(L_{S}+L_{P})\right]\\
 & \leq2\mathbb{E}\left[r^{2}\right]+2\log(L_{S}+L_{P})\mathbb{E}\left[r\right],\end{align*}
where we used the fact that  $\beta(L_{S}+L_{P})\leq2$
and $r\log r\leq r^{2}.$ We can write \[
\mathbb{E}\left[r\right]=\frac{1}{p},\mathbb{E}\left[r^{2}\right]=\mbox{Var}(r)+\mathbb{E}\left[r\right]^{2}=\frac{2-p}{p^{2}}.\]
Also, we use $\log(L_{S}+L_{P})\leq\log 2L_{S}\leq 2\log\frac{1}{\beta}$
and find that \begin{align*}
\mathbb{E}\left[\beta|F_{j}|\log|F_{j}|\right] & \leq\frac{4-2p}{p^{2}}+\frac{4}{p}\log\frac{1}{\beta}\leq16+8\log\frac{1}{\beta},\end{align*}
where we used the fact that  $\frac{4-2p}{p^{2}}\leq16$ and $\frac{4}{p} \leq 8$ for  $p\geq\frac{1}{2}$  (Notice that $p\rightarrow 1$, as $\beta \rightarrow 0$.).
\section*{Appendix II}
Recall that $\delta_{j}$ is the number of deleted bits from the substring
of $X$ between pivots $P_{i_{j}}$ and $P_{i_{j+1}}.$ Let us denote
by $\mathcal{L}_{P}$ the set of indices $l$ for which $P_{l}$ appears
between $P_{i_{j}}$ and $P_{i_{j+1}}$ for some $j\in\mathcal{H}.$
Similarly, let $\mathcal{L}_{S}$ denote the set of indices $l$ for
which $S_{l}$ appears between $P_{i_{j}}$ and $P_{i_{j+1}}$ for
some $j\in\mathcal{H}.$ Let $\delta_{P_{l}}$ denote the number of
deleted bits from $P_{l}$ and $\delta_{S_{l}}$ denote the number
of deleted bits from $S_{l}$. We can write \[
\sum_{j\in\mathcal{H}}\delta_{j}=\sum_{l\in\mathcal{L}_{P}}\delta_{P_{l}}+\sum_{l\in\mathcal{L}_{S}}\delta_{S_{l}}.\]
 Notice that the length of the interval that $\delta_{P_{l}}$ takes
values from is $L_{P}=O(\log\frac{1}{\beta})$ and the length of the
interval that $\delta_{S_{l}}$ takes values from is $L_{S}=\frac{1}{\beta}.$
Next, by application of Theorem \ref{Hoeffding_2} we can write \begin{align*}
 & \Pr\left\{ \left|\sum_{j\in\mathcal{H}}\delta_{j}-\mathbb{E}\left[\sum_{j\in\mathcal{H}}\delta_{j}\right]\right|=o(\beta)k\right\} \geq\\
 & \Pr\left\{ \left|\sum_{l\in\mathcal{L}_{P}}\delta_{P_{l}}-\mathbb{E}\left[\sum_{l\in\mathcal{L}_{P}}\delta_{P_{l}}\right]\right| \right.+ \\
& \left. \left|\sum_{l\in\mathcal{L}_{S}}\delta_{S_{l}}-\mathbb{E}\left[\sum_{l\in\mathcal{L}_{S}}\delta_{S_{l}}\right]\right|=o(\beta)k\right\} =\\
 & \Pr\left\{ \left|\sum_{l\in\mathcal{L}_{P}}\delta_{P_{l}}-\mathbb{E}\left[\sum_{l\in\mathcal{L}_{P}}\delta_{P_{l}}\right]\right|=o(\beta)k\right\} \cdot\\
 & \qquad\qquad\Pr\left\{ \left|\sum_{l\in\mathcal{L}_{S}}\delta_{S_{l}}-\mathbb{E}\left[\sum_{l\in\mathcal{L}_{S}}\delta_{S_{l}}\right]\right|=o(\beta)k\right\} =\\
 & \Pr\left\{ \left|\sum_{l\in\mathcal{L}_{P}}\delta_{P_{l}}-\mathbb{E}\left[\sum_{l\in\mathcal{L}_{P}}\delta_{P_{l}}\right]\right|=\frac{o(\beta)k}{|\mathcal{L}_{P}|}|\mathcal{L}_{P}|\right\} \cdot\\
 & \qquad\qquad\Pr\left\{ \left|\sum_{l\in\mathcal{L}_{S}}\delta_{S_{l}}-\mathbb{E}\left[\sum_{l\in\mathcal{L}_{S}}\delta_{S_{l}}\right]\right|=\frac{o(\beta)k}{|\mathcal{L}_{S}|}|\mathcal{L}_{S}|\right\} \geq 
 \end{align*}
 \begin{align*}
 & \left(1-2\exp(-\frac{2o(\beta^{2})k^{2}|\mathcal{L}_{P}|}{|\mathcal{L}_{P}|^{2}L_{P}^{2}})\right)\cdot\\
 & \hspace{1.0in}\left(1-2\exp(-\frac{2o(\beta^{2})k^{2}|\mathcal{L}_{S}|}{|\mathcal{L}_{S}|^{2}L_{S}^{2}})\right)=\\
 & \left(1-2\exp(-\frac{2o(\beta^{2})k^{2}}{|\mathcal{L}_{P}|L_{P}^{2}})\right)\cdot \\
 & \hspace{1.2in} \left(1-2\exp(-\frac{2o(\beta^{2})k^{2}}{|\mathcal{L}_{S}|L_{S}^{2}})\right)\geq\\
 & \left(1-2\exp(-\frac{2o(\beta^{2})k}{L_{P}^{2}})\right)\cdot\left(1-2\exp(-\frac{2o(\beta^{2})k}{L_{S}^{2}})\right)\geq\\
 & \left(1-2\exp(-\frac{2o(\beta^{2})\beta}{O(\log^2\frac{1}{\beta})}n)\right)\cdot\left(1-2\exp(-2\beta^{3}o(\beta^{2})n)\right)=\\
 & (1-2^{-\Omega(n)})\cdot(1-2^{-\Omega(n)})=1-2^{-\Omega(n)},\end{align*}
 where in our derivation we used the fact that $|\mathcal{L}_{P}|\leq k$ 
and $|\mathcal{L}_{S}|\leq k,$ since $\mathcal{L}_{P}$ and $\mathcal{L}_{S}$ are subsets of $\{1,\cdots,k-1\}.$

\begin{IEEEbiographynophoto}
{S. M. Sadegh Tabatabaei Yazdi} is a senior engineer at Qualcomm Research and Development center in San Diego, CA. Prior to joining Qualcomm, from August 2011 to July 2012, he was a Postdoctoral Scholar at the Electrical Engineering department at UCLA where he was working on the design of optimal LDPC decoders for erroneous hardware and on the design of optimal coding schemes for synchronization channels. Dr. Tabatabaei Yazdi received his PhD degree from Texas A\&M University in August 2011 and his Masters degree from University of Michigan, Ann Arbor in December 2007. The focus of his PhD and Masters research was on the design of low complexity and optimal network codes for different topologies of wired and wireless networks. He also received his Bachelors degree in Electrical Engineering from Sharif University of Technology, Tehran, Iran, in June 2006.
\end{IEEEbiographynophoto}
\begin{IEEEbiographynophoto}
{Lara Dolecek}
 (S'05-- M'10--SM'13) is an Assistant Professor  with the Electrical Engineering Department at the University of California, Los Angeles (UCLA) where she is the director of the Laboratory for Robust Information Systems. She holds a B.S. (with honors), M.S. and Ph.D. degrees in Electrical Engineering and Computer Sciences, as well as an M.A. degree in Statistics, all from the University of California, Berkeley. She received the 2007 David J. Sakrison Memorial Prize for the most outstanding doctoral research in the Department of Electrical Engineering and Computer Sciences at UC Berkeley. Prior to joining UCLA, she was a postdoctoral researcher with the Laboratory for Information and Decision Systems at the Massachusetts Institute of Technology. She received  Intel Early Career Faculty Award, University of California Faculty Development Award, and Okawa Research Grant all in 2013, NSF CAREER Award in 2012, and Hellman Fellowship Award in 2011. She is an Associate Editor for IEEE Transactions on Communications and for IEEE Communication Letters and is the lead guest editor for JSAC special issue on emerging data storage. Her research interests span coding and information theory, graphical models, statistical algorithms, and computational methods, with applications to emerging systems for data storage, processing, and communication.
\end{IEEEbiographynophoto}
\end{document}